\pdfoutput=1

\documentclass{article}

\usepackage[a4paper,left=4cm]{geometry}	
\usepackage{lmodern,microtype}	
\usepackage{tikz}				
\usepackage[english]{babel}		
\usepackage{color}				

\definecolor{darkblue}{rgb}{0,0,.8}

\usepackage[nobreak]{cite} 

\usepackage{hyperref}
\hypersetup{
    colorlinks,
    citecolor=red,
    filecolor=black,
    linkcolor=darkblue,
    urlcolor=black
}

\usepackage{amsmath,amssymb,amsthm}
\usepackage{bm}

\usepackage[capitalise, noabbrev]{cleveref}

\newtheorem{theorem}{Theorem}[section]
\newtheorem{lemma}[theorem]{Lemma}
\newtheorem{proposition}[theorem]{Proposition}
\newtheorem{corollary}[theorem]{Corollary}
\newtheorem{conjecture}[theorem]{Conjecture}

\renewcommand{\i}{\text{i}}

\newcommand{\diff}{\textup{d}}

\newcommand{\Q}{\mathfrak{Q}}
\newcommand{\q}{\mathfrak{q}}

\newcommand{\ket}[1]{|{#1}\rangle}
\newcommand{\bra}[1]{\langle {#1}|}
\newcommand{\uu}{\uparrow}
\newcommand{\dd}{\downarrow}
\newcommand{\vth}[3]{\vartheta_{{#1}}({#2},{#3})}

\title{\Large\bf On the transfer matrix of the supersymmetric eight-vertex model. II. Open boundary conditions}
\author{\normalsize \textsc{Christian Hagendorf} and \textsc{Jean Li\'enardy} \\
\medskip\\
{\normalsize
  \begin{minipage}{\textwidth}
  \begin{center}
  \textit{
   Universit\'e catholique de Louvain\\
  Institut de Recherche en Math\'ematique et Physique\\
  Chemin du Cyclotron 2, 1348 Louvain-la-Neuve, Belgium} \\
  \bigskip
  \href{mailto:christian.hagendorf@uclouvain.be}{\normalsize 
\texttt{christian.hagendorf@uclouvain.be}}, \href{mailto:jean.lienardy@uclouvain.be}{\normalsize 
\texttt{jean.lienardy@uclouvain.be}}
  \end{center}
  \end{minipage}
}
}

\date{\normalsize{\today}}

\begin{document}
\maketitle

\begin{abstract}
The transfer matrix of the square-lattice eight-vertex model on a strip with $L\geqslant 1$ vertical lines and open boundary conditions is investigated. It is shown that for vertex weights $a,b,c,d$ that obey the relation $(a^2+ab)(b^2+ab)=(c^2+ab)(d^2+ab)$ and appropriately chosen $K$-matrices $K^\pm$ this transfer matrix possesses the remarkably simple, non-degenerate eigenvalue $\Lambda_L = (a+b)^{2L}\,\text{tr}(K^+K^-)$. For positive vertex weights, $\Lambda_L$ is shown to be the largest transfer-matrix eigenvalue. The corresponding eigenspace is equal to the space of the ground states of the Hamiltonian of a related XYZ spin chain. An essential ingredient in the proofs is the supersymmetry of this Hamiltonian.
\end{abstract}

\section{Introduction}

In this article, we continue our investigation of the eight-vertex model whose vertex weights $a,b,c,d$ are non-zero and obey the relation 
\begin{equation}
  \label{eqn:CL8V}
  (a^2+ab)(b^2+ab)=(c^2+ab)(d^2+ab).
\end{equation}
In 2001, Stroganov \cite{stroganov:01} studied this special case of the eight-vertex model with periodic boundary conditions. He conjectured that, for each $n\geqslant 0$, its transfer matrix with an odd number $L=2n+1$ of vertical lines possesses the remarkably simple doubly-degenerate eigenvalue $\Theta_n= (a+b)^{2n+1}$.  Stroganov's conjecture led several authors to investigate the case \eqref{eqn:CL8V}, which has revealed interesting relations between the eight-vertex model and a variety of mathematical structures and topics. Amongst these are enumerative combinatorics \cite{razumov:10,mangazeev:10,zinnjustin:13}, functional equations \cite{baxter:89,fabricius:05,rosengren:16} and solutions to the Painlev\'e VI equation \cite{bazhanov:05,bazhanov:06,rosengren:13,rosengren:13_2,rosengren:14,rosengren:15}. Furthermore, a relation to supersymmetry was established in \cite{fendley:10,hagendorf:12,hagendorf:13}. Therefore, we refer to the case \eqref{eqn:CL8V} as the supersymmetric eight-vertex model. In \cite{hagendorf:18}, we used the supersymmetry to prove Stroganov's conjecture. Furthermore, we showed that $\Theta_n$ is the largest eigenvalue of the transfer matrix for positive vertex weights. The present article aims to extend our work to the eight-vertex model on a strip.

As in our previous work on periodic boundary conditions, we exploit a well-known relation between the eight-vertex model and  the XYZ quantum spin chain. For a spin chain with $L\geqslant 1$ sites and open boundary conditions, its Hamiltonian is given by
\begin{subequations}
\label{eqn:SUSYXYZHamiltonian}
\begin{equation}
  \label{eqn:XYZHamiltonian}
  H_{\text{\tiny XYZ}} = -\frac{1}{2}\sum_{j=1}^{L-1}\left(J_1 \sigma^1_j\sigma^1_{j+1}+J_2 \sigma^2_j\sigma^2_{j+1}+J_3\sigma^3_j\sigma^3_{j+1}\right) + (h_{\text{\tiny B}}^-)_1+ (h_{\text{\tiny B}}^+)_L.
\end{equation}
(By convention, for $L=1$, the bulk interaction term is absent and the Hamiltonian is given by the sum $H_{\text{\tiny XYZ}} = h_{\text{\tiny B}}^++h_{\text{\tiny B}}^-$.) Here, $\sigma^1,\sigma^2,\sigma^3$ denote the Pauli matrices. The constants $J_1,J_2,J_3$ are the spin chain's anisotropy parameters. We focus on the case where they are given by
\begin{equation}
  \label{eqn:APCL}
  J_1 = 1+\zeta,\quad J_2 = 1-\zeta,\quad J_3 = \frac{1}{2}(\zeta^2-1),
\end{equation}
with a real parameter $\zeta$. The terms $h_{\text{\tiny B}}^\pm$ describe the interactions of the first and last spins with boundary magnetic fields. We consider the boundary terms
\begin{equation}
  \label{eqn:hBPM1}
  h_\text{\tiny B}^+=h_\text{\tiny B}^- = \sum_{\alpha=1}^3 \lambda_\alpha\sigma^\alpha,
\end{equation}
where
\begin{equation}
  \label{eqn:hBPM2}
  \lambda_1 = -\frac{(1+\zeta)\mathop{\text{Re}} y}{1+|y|^2}, \quad \lambda_2 = -\frac{(1-\zeta)\mathop{\text{Im}} y}{1+|y|^2},\quad \lambda_3 = \left(\frac{\zeta^2-1}{4}\right)\left(\frac{1-|y|^2}{1+|y|^2}\right),
\end{equation}
\end{subequations}
and $y$ is a complex number.

\paragraph{Supersymmetry.} We show that similarly to the case of periodic boundary conditions \cite{hagendorf:13}, the Hamiltonian \eqref{eqn:SUSYXYZHamiltonian} is supersymmetric: Up to a constant shift and a rescaling, it can be written as the anticommutator of a nilpotent operator, the \textit{supercharge}, and its adjoint. The supersymmetry implies that the Hamiltonian may have special eigenstates called \textit{supersymmetry singlets}. They are annihilated by both the supercharge and its adjoint. If they exist, then they are the Hamiltonian's ground states, and hence of physical interest \cite{witten:82}.

Therefore, we wish to find the pairs $(\zeta,y)$ for which the Hamiltonian possesses supersymmetry singlets. We shall see later that it is sufficient consider $0\leqslant \zeta \leqslant 1$. For $\zeta=0$, the Hamiltonian \eqref{eqn:SUSYXYZHamiltonian} reduces to a XXZ spin-chain Hamiltonian. In \cite{hagendorf:17}, we showed that the space of its ground states is the space of supersymmetry singlets if and only if $y=0$. Furthermore, for $\zeta=1$, the Hamiltonian greatly simplifies and it is trivial to find its ground states (whether they are supersymmetry singlets or not). Hence, we focus on $0<\zeta<1$.
\begin{theorem}
\label{thm:MainTheorem1}
For each $L\geqslant 1$ and $0< \zeta < 1$, the space of the ground states of the Hamiltonian \eqref{eqn:SUSYXYZHamiltonian} is equal to the space of supersymmetry singlets if and only if $y$ is a solution of the polynomial equation
\begin{equation}
  \label{eqn:YEqn}
  \zeta(1+y^4) -(3-\zeta^2)y^2 = 0 .
\end{equation}
This space is one-dimensional, and the corresponding ground-state eigenvalue is given by
\begin{equation}
    \label{eqn:E0}
    E_0 = -\frac{(L-1)(3+\zeta^2)}{4}-\frac{(1+\zeta)^2}{2}.
  \end{equation}
\end{theorem}

\paragraph{The transfer-matrix eigenvalue.} For all $y$, the transfer matrix of the supersymmetric eight-vertex model on a strip with $L$ vertical lines commutes with the Hamiltonian \eqref{eqn:SUSYXYZHamiltonian} if
\begin{equation}
  \label{eqn:Zabcd}
  \zeta = \frac{cd}{ab},
\end{equation}
and if the boundary conditions of the strip, encoded in the so-called $K$-matrices, are chosen in accordance with the boundary terms $h_\text{\tiny B}^\pm$ \cite{sklyanin:88}. These $K$-matrices are given by
\begin{align}
\label{eqn:DefKpm}
\begin{split}
K^- &= 
1 + \frac{2 \mathop{\text{Re}} y}{1+|y|^2}\frac{ab+cd}{ac+bd}\sigma^1  + \frac{2 \mathop{\text{Im}} y}{1+|y|^2}\frac{ab-cd}{ac-bd}\sigma^2 +\frac{1-|y|^2}{1+|y|^2}\frac{b^2-d^2}{2 ab + b^2+d^2}\sigma^3,
\\
 K^+ &=
 1 + \frac{2 \mathop{\text{Re}} y}{1+|y|^2}\frac{ab+cd}{ad+bc}\sigma^1 +\frac{2 \mathop{\text{Im}} y}{1+|y|^2}\frac{ab-cd}{bc-ad}\sigma^2+\frac{1-|y|^2}{1+|y|^2}\frac{b^2-c^2}{2 ab + b^2+c^2}\sigma^3.
\end{split}
\end{align}%
By \cref{thm:MainTheorem1}, if $y$ is a solution of \eqref{eqn:YEqn}, then the space of supersymmetry singlets is necessarily an eigenspace of the transfer matrix of the supersymmetric eight-vertex model on a strip with these $K$-matrices. 
\begin{theorem}
  \label{thm:MainTheorem2}
  Let $L\geqslant 1$, $0<\zeta<1$ and $y$ be a solution of \eqref{eqn:YEqn}, then the transfer matrix of the supersymmetric eight-vertex model on a strip with $L$ vertical lines and the $K$-matrices \eqref{eqn:DefKpm} possesses the non-degenerate eigenvalue
  \begin{equation}
    \Lambda_L = (a+b)^{2L}\,\textup{tr}(K^+K^-).
    \label{eqn:DefLambda}
  \end{equation}
  The corresponding eigenspace is the space of the supersymmetry singlets of the XYZ Hamiltonian \eqref{eqn:SUSYXYZHamiltonian} with $\zeta$ given by \eqref{eqn:Zabcd}.
\end{theorem}
Finally, using the Perron-Frobenius theorem, we prove the following:
\begin{theorem}
   \label{thm:MainTheorem3}
Let $L\geqslant 1$, $0<\zeta<1$ and $y$ be a solution of \eqref{eqn:YEqn}. If $a,b,c,d>0$, then $\Lambda_L$ is the largest eigenvalue of the transfer matrix of the supersymmetric eight-vertex model on a strip with $L$ vertical lines and the $K$-matrices \eqref{eqn:DefKpm}.
\end{theorem}

Theorems \ref{thm:MainTheorem1}, \ref{thm:MainTheorem2} and \ref{thm:MainTheorem3} are the main results of this article. We stress that supersymmetry is an essential ingredient of their proofs. Indeed, we do not use traditional methods that allow one to analyse the spectrum of the transfer matrix of the eight-vertex model on a strip. Two examples of these methods are the off-diagonal Bethe ansatz \cite{cao:13_2} and the quantum separation of variables method \cite{faldella:14}. To our best knowledge, finding an explicit expression of the largest transfer-matrix eigenvalue for finite $L$ with these methods remains a challenge, even if the vertex weights obey \eqref{eqn:CL8V}.

The layout of this article is similar to \cite{hagendorf:18}. In \cref{sec:SUSY}, we study the Hamiltonian \eqref{eqn:SUSYXYZHamiltonian} and its supersymmetry. In particular, we investigate the existence of supersymmetry singlets and prove the \cref{thm:MainTheorem1}. We compute the action of the transfer matrix of the supersymmetric eight-vertex model with open boundary conditions on the space of supersymmetry singlets in \cref{sec:8V}. To this end, we recall the construction of the transfer matrix and its relation to the Hamiltonian of the XYZ spin chain. Moreover, we establish a commutation relation between the transfer matrix and the supercharge of the spin chain. This relation allows us to prove \cref{thm:MainTheorem2}. 
 In \cref{sec:LargestEV}, we analyse the positivity of the transfer matrix and use the Perron-Frobenius theorem to prove \cref{thm:MainTheorem3}. We present our conclusions in \cref{sec:Conclusion} and conjecture a generalisation of $\Lambda_L$ for the inhomogeneous eight-vertex model.

\section{Supersymmetry}

\label{sec:SUSY}
In this section, we investigate the supersymmetry of the XYZ Hamiltonian \eqref{eqn:SUSYXYZHamiltonian}. We start this investigation in \cref{sec:Parameters} with a short discussion on the transformations of the Hamiltonian's parameters. In \cref{sec:DefSUSY}, we define a supercharge for the Hamiltonian. In \cref{sec:Parameterisation}, we introduce a basis of the spin Hilbert space in which the action of this supercharge is simple. We use this basis in \cref{sec:Cohomology} to compute the (co)homology of the supercharge and its adjoint. In \cref{sec:Singlets}, we discuss the absence or existence of supersymmetry singlets of the Hamiltonian.

\subsection{Parameter range}
\label{sec:Parameters}

\paragraph{Notation.} Let us recall basic notations and conventions (which are similar to \cite{hagendorf:17,hagendorf:18}). We use the notation $V=\mathbb C^2$ for the Hilbert space of a spin $1/2$. A basis of this Hilbert space is 
\begin{equation}
  |{\uparrow}\rangle
  =\begin{pmatrix}
    1\\ 0
  \end{pmatrix},
  \quad 
  |{\downarrow}\rangle
  =\begin{pmatrix}
    0\\ 1
  \end{pmatrix}.
\end{equation}
The Hilbert space of a spin chain with $L\geqslant 1$ sites is given by $V^L = V_1\otimes V_2\otimes \cdots \otimes V_L$ where $V_j =V$ is a copy of the single-spin Hilbert space associated to the site $j$. A basis of $V^L$ is given by the states
\begin{equation}
  \label{eqn:BasisStates}
  |s_1s_2\cdots s_L\rangle = |s_1\rangle \otimes |s_2\rangle \otimes \cdots \otimes |s_L\rangle,
\end{equation}
where $s_j\in \{\uparrow,\downarrow\}$ for each $j=1,\dots,L$. Furthermore, we denote by $\langle \psi|\psi'\rangle$ the canonical (complex) scalar product of any two states $|\psi\rangle,\,|\psi'\rangle \in V^L$, where $\langle \psi| = |\psi\rangle^\dagger$. Finally, we denote by $\sigma_j^1,\sigma^2_j,\sigma^3_j$ the Pauli matrices 
\begin{equation}\label{eqn:PauliMat}
 \sigma^1 = \begin{pmatrix} 0&1\\ 1&0 \end{pmatrix}, \quad 
 \sigma^2 =\begin{pmatrix}0 &-\i\\ \i&0 \end{pmatrix},\quad
  \sigma^3 = \begin{pmatrix}1&0\\ 0&-1 \end{pmatrix},
\end{equation}
acting on the $j$-th factor of the basis states \eqref{eqn:BasisStates}.

\paragraph{Transformation of the parameters.} We analyse the transformation behaviour of the Hamiltonian \eqref{eqn:SUSYXYZHamiltonian} under spin rotations. To this end, we introduce the operators
\begin{equation}
  \mathcal R^\alpha(\theta) = \exp\left(\frac{\i \theta}{2}(\sigma_1^\alpha+\cdots+\sigma_L^\alpha)\right), \quad \alpha =1,2,3.
\end{equation}
We write $H_{\textup{\tiny XYZ}}=H_{\textup{\tiny XYZ}}(\zeta,y)$ to stress the dependence of the Hamiltonian on $\zeta$ and $y$. For each $L\geqslant 1$, it transforms under rotations by the angle $\theta=\pi/2$ according to
\begin{align}
  \begin{split}
  \label{eqn:PiOverTwo}%
  \mathcal R^1(\pi/2)H_{\textup{\tiny XYZ}}(\zeta,y)\mathcal R^1(-\pi/2) &= \left(\frac{1+\zeta}{2}\right)^2H_{\textup{\tiny XYZ}}\left(\frac{3-\zeta}{1+\zeta},\frac{y-\i}{1-\i y}\right),\\
  \mathcal R^2(\pi/2)H_{\textup{\tiny XYZ}}(\zeta,y)\mathcal R^2(-\pi/2) &= \left(\frac{1-\zeta}{2}\right)^2H_{\textup{\tiny XYZ}}\left(\frac{\zeta+3}{\zeta-1},\frac{1+y}{1-y}\right),\\
   \mathcal R^3(\pi/2)H_{\textup{\tiny XYZ}}(\zeta,y)\mathcal R^3(-\pi/2) &= H_{\textup{\tiny XYZ}}\left(-\zeta,-\i y\right).
    \end{split}
\end{align}
Two successive applications of \eqref{eqn:PiOverTwo} lead to the following transformations under rotations by the angle $\theta=\pi$:
\begin{align}
  \begin{split}
  \label{eqn:Pi}%
  \mathcal R^1(\pi)H_{\textup{\tiny XYZ}}(\zeta,y) \mathcal R^1(-\pi)  &=  H_{\textup{\tiny XYZ}}\left(\zeta,y^{-1}\right),\\
  \mathcal R^2(\pi)H_{\textup{\tiny XYZ}}(\zeta,y) \mathcal R^2(-\pi)  &=  H_{\textup{\tiny XYZ}}\left(\zeta,-y^{-1}\right),\\
  \mathcal R^3(\pi)H_{\textup{\tiny XYZ}}(\zeta,y)\mathcal R^3(-\pi) &=  H_{\textup{\tiny XYZ}}(\zeta,-y).
  \end{split}
\end{align}%
The transformations \eqref{eqn:PiOverTwo} and \eqref{eqn:Pi} are unitary. Therefore, they do not change the spectrum of the Hamiltonian. Moreover, they allow us to transform a Hamiltonian with arbitrary parameters $\zeta$ and $y$ to a Hamiltonian whose parameters are restricted to a domain defined by the inequalities
\begin{equation}
  0\leqslant \zeta \leqslant 1,\, 0\leqslant |y|\leqslant 1,\,\text{Re}\,y\geqslant 0.
\end{equation}

As stated in the introduction, we investigate for which pairs $(\zeta,y)$ the ground states of the Hamiltonian are supersymmetry singlets. The case $\zeta=0$ was addressed in \cite{hagendorf:17}. Furthermore, the case $\zeta=1$ is trivial. Indeed, in this case, the Hamiltonian is
\begin{equation}
  H_{\textup{\tiny XYZ}} = - \sum_{j=1}^{L-1} \sigma_j^1\sigma_{j+1}^1 - \frac{2\textup{Re}\, y}{1+|y|^2}\left(\sigma_1^1+\sigma_L^1\right).
\end{equation}
Its ground states are easily found. Therefore, we focus on $0<\zeta<1$. We often focus on the case where $(\zeta,y)$ belongs to the domain
\begin{equation}
  \mathbb D =\{(\zeta,y): 0< \zeta < 1
  ,\, 0\leqslant |y|\leqslant 1,\,\text{Re}\,y\geqslant 0\}.
  \label{eqn:ParameterRange}
\end{equation}

\subsection{The supersymmetry}
\label{sec:DefSUSY}

\paragraph{Local supercharges and supercharges.} The construction of the supersymmetry for the XYZ spin chain is based on operators $\q: V\to V\otimes V$ that we call \textit{local supercharges}. We consider local supercharges with the property
\begin{equation}
  (\q \otimes 1 - 1 \otimes \q)\q|\psi\rangle = |\chi\rangle \otimes |\psi\rangle - |\psi\rangle \otimes |\chi\rangle,
  \label{eqn:QuasiCoassociativity}
\end{equation}
for all $|\psi\rangle \in  V$. Here $|\chi\rangle \in  V \otimes V$ is a fixed state. If $|\chi\rangle = 0$ then \eqref{eqn:QuasiCoassociativity} reduces to
\begin{equation}
  (\q \otimes 1 - 1 \otimes \q)\q=0.
\end{equation}
We call a local supercharge with this property \textit{coassociative}. Coassociative local supercharges allow us to construct supercharges for open spin chains \cite{hagendorf:17}. To see this, we consider the local operators $\q_j,\, j=1,\dots,L,$ on $V^L$ that are given by
\begin{equation}
  \q_j = \underset{j-1}{\underbrace{1\otimes \cdots \otimes 1}}\otimes \q \otimes \underset{L-j}{\underbrace{1\otimes \cdots \otimes 1}}.
\end{equation}
They map $V^L$ to $V^{L+1}$. Using these operators, we define\footnote{As in \cite{hagendorf:18}, we use the symbol $\Q$ irrespectively of the space on which the supercharge acts. This allows us to simplify the notation. If necessary, we indicate the space by writing $\Q:V^L\to V^{L+1}$ (or $\Q:V^{L-1}\to V^L$ etc.), or by explicitly stating that $\Q$ acts on $V^L$ (or $V^{L-1}$ etc.). We use the same convention for other operators, such as the operators that derive from the supercharge.} for each $L\geqslant 1$ the supercharge $\Q:V^L \to V^{L+1}$ as the linear combination
\begin{equation}
  \Q = \sum_{j=1}^L (-1)^j \q_j.
  \label{eqn:DefQ}
\end{equation}
For each $L\geqslant 2$, the adjoint supercharge $\Q^\dagger:V^{L}\to V^{L-1}$ is defined by means of the scalar product of the spin-chain Hilbert space: We set $\langle \psi|\Q^\dagger|\phi\rangle = (\langle \phi|\Q|\psi\rangle )^\ast$ for each $|\phi\rangle \in V^L,\,|\psi\rangle \in V^{L-1}$. The operators $\Q$ and $\Q^\dagger$ are nilpotent,
\begin{equation}
  \Q^2 = 0,\quad  (\Q^\dagger)^2=0,
\end{equation}
if and only if the local supercharge $\q$ is coassociative. This means that the mappings $\Q^2:V^L \to V^{L+2}, \, L\geqslant 1,$ and $(\Q^\dagger)^2:V^L\to V^{L-2},\, L\geqslant 3,$ yield zero on every state of $V^L$. 

\paragraph{Hamiltonian.}
We use $\Q$ and $\Q^\dagger$ to define a Hamiltonian $H$. For $L=1$, it is given by $H = \Q^\dagger \Q$. For $L\geqslant 2$, it is the anticommutator
\begin{equation}
  H = \Q\Q^\dagger + \Q^\dagger \Q.
  \label{eqn:HFromQ}
\end{equation}
It was shown in \cite{hagendorf:17} that this Hamiltonian is given by a sum of terms that describe interactions between nearest neighbours and boundary terms. Furthermore, the nilpotency of $\Q$ and $\Q^\dagger$ implies the commutation relations
\begin{subequations}
\label{eqn:CommHQ}
\begin{equation}
  H\Q = \Q H,\quad \text{on } V^L,\, L\geqslant 1,
\end{equation}
and
\begin{equation}
  H\Q^\dagger = \Q^\dagger H,\quad \text{on } V^L,\, L\geqslant 2.
\end{equation}
\end{subequations}
Hence, the system described by the Hamiltonian $H$ is supersymmetric \cite{witten:82}. We note, however, that the Hamiltonians on the left- and right-hand side of these relations act on the Hilbert spaces of spin chains whose lengths differ by one.

\paragraph{A local supercharge for the XYZ spin chain.} 
We now construct a local supercharge that allows us to investigate the Hamiltonian \eqref{eqn:SUSYXYZHamiltonian}. To this end, we define three local supercharges that satisfy \eqref{eqn:QuasiCoassociativity}. First, we introduce the operator  $\q_\phi$ that acts on $|\psi\rangle \in V$ according to
\begin{equation}
  \q_\phi|\psi\rangle = |\phi\rangle \otimes |\psi\rangle +|\psi\rangle \otimes |\phi\rangle.
  \label{eqn:DefQPhi}
\end{equation}
 Here $|\phi\rangle \in V$ is a fixed state. Indeed, $\q_\phi$ obeys \eqref{eqn:QuasiCoassociativity} with $|\chi\rangle = |\phi\rangle \otimes |\phi\rangle$. Hence, if $|\phi\rangle$ is non-zero then the local supercharge $\q_\phi$ is not coassociative. Second, we define $\q^\uparrow$ and $\q^\downarrow$ through the following action on the basis vectors of $V$ \cite{hagendorf:13}:
\begin{subequations}
\label{eqn:XYZSC}%
\begin{alignat}{4}
  &\q^\uparrow|{\uparrow}\rangle = 0,& \quad &  \q^\uparrow|{\downarrow}\rangle = |{\uparrow\uparrow}\rangle-\zeta|{\downarrow\downarrow}\rangle,\\
 & \q^\downarrow|{\downarrow}\rangle = 0,  &\quad  &\q^\downarrow |{\uparrow}\rangle = |{\downarrow\downarrow}\rangle-\zeta|{\uparrow\uparrow}\rangle.
\end{alignat}%
\end{subequations}%
One checks that both $\q^\uparrow$ and $\q^\downarrow$ obey \eqref{eqn:QuasiCoassociativity} with the vectors $|\chi\rangle = -\zeta|{\uparrow\uparrow}\rangle$ and $|\chi\rangle = -\zeta|{\downarrow\downarrow}\rangle$, respectively. Hence, these operators are not coassociative for non-zero $\zeta$.

We use the three local supercharges $\q^\uparrow,\q^\downarrow$ and $\q_\phi$ to define the linear combination
\begin{subequations}
\label{eqn:XYZLocalSupercharge}
\begin{equation}
  \q = (1-y^2 \zeta)\q^\uparrow+y(y^2-\zeta)\q^\downarrow+\q_\phi,
\end{equation}
where $\ket{\phi}$ is given by
\begin{equation}
  |\phi\rangle  = y(y^2\zeta-1)|{\uparrow}\rangle+(\zeta-y^2)|{\downarrow}\rangle,
\end{equation}
\end{subequations}%
and $y$ is a complex number.
A straightforward calculation shows that $\q$ is coassociative for all $\zeta$ and $y$.

In the next proposition, we prove that the Hamiltonian \eqref{eqn:SUSYXYZHamiltonian} of the XYZ spin chain is supersymmetric, up to a rescaling and to adding a multiple of the identity matrix.
\begin{proposition}
\label{prop:HSUSYHXYZ}
  For each $L\geqslant 1$, the Hamiltonian \eqref{eqn:HFromQ} constructed from the local supercharge \eqref{eqn:XYZLocalSupercharge} is
  \begin{equation}
  \label{eqn:RelationsHam}
   H = x\left(H_{\textup{\tiny XYZ}} + \frac{(L-1)(\zeta^2+3)}{4}+2\lambda_0\right)\!,\end{equation}
  where $H_{\textup{\tiny XYZ}}$ is defined in \eqref{eqn:SUSYXYZHamiltonian}. We have
  \begin{equation}
 \lambda_0 = \frac{1+3\zeta^2}{4}- \frac{(\zeta^2-1)((3 + \zeta^2)|y|^2 - 4 \zeta  \mathop{\textup{Re}}(y^2)  )}{2 (1+|y|^4+(\zeta^2-1)|y|^2-2 \zeta \mathop{\textup{Re}}(y^2) )},
\end{equation}
and
\begin{equation}
x=(1+|y|^2)(1+|y|^4+(\zeta^2-1)|y|^2- 2\zeta \mathop{\textup{Re}}(y^2) ).
\end{equation}
  \begin{proof}
    The proof is a straightforward calculation that follows \cite{hagendorf:17}.
  \end{proof} 
\end{proposition}

\subsection{Theta-function parameterisation}

\label{sec:Parameterisation}

In this section, we introduce a parameterisation of the points $(\zeta,y)\in \mathbb D$ in terms of Jacobi theta functions. We employ this theta-function parameterisation to define a new basis of the spin Hilbert space. The action of the local supercharge \eqref{eqn:XYZLocalSupercharge} on the basis states yields simple results. 

\paragraph{Parameterisation.} We use the classical notations $\vartheta_j(u,p),\,1\leqslant j \leqslant 4$ and definitions for the Jacobi theta functions \cite{whittaker:27,gradshteyn:07}. We only consider a real elliptic nome $p$ with
\begin{equation}
  0<p<1.
\end{equation}
Let us write $p=e^{-s},\,s>0.$ We define the rectangle $\mathbb R_p = \{z\in \mathbb C: 0 \leqslant \text{Re}\,z \leqslant \pi/2,\, -s/2 \leqslant \text{Im}\, z \leqslant s/2\}$, and the domain
\begin{equation}
  \bar {\mathbb D} = \{(p,t): 0<p<1, \,t \in \mathbb R_p\}.
\end{equation}
The parameterisation of $(\zeta,y)\in \mathbb D$ in terms of $(p,t)\in \bar{\mathbb D}$ is given by
\begin{equation} 
\label{eqn:JacobiThetaParam}
\zeta = \left(\frac{\vartheta_1(2\pi/3,p^2)}{\vartheta_4(2\pi/3,p^2)}\right)^2, \quad y = \frac{\vartheta_1(t,p^2)}{\vartheta_4(t,p^2)}.
\end{equation}
It has the following property:

\begin{proposition} The parameterisation \eqref{eqn:JacobiThetaParam} defines a bijection between $\bar{\mathbb D}$ and $\mathbb D$.
\begin{proof}
We only sketch the proof. First, we note that $\zeta$ is a monotone function of $p$. Second, as a function of $t$, $y$ is the Jacobi elliptic function $\textup{sn}$, up to a rescaling of its argument and a constant factor. The bijectivity can be established with the help of the monotonicity and the conformal mapping properties of $\textup{sn}$ \cite{nehari:82}.
\end{proof}
\end{proposition}

In the remainder of this section, we implicitly assume the parameterisation \eqref{eqn:JacobiThetaParam}.

\paragraph{Basis states.} In addition to the parameterisation, we introduce the states and dual states
\begin{align}
  \label{eqn:DefVSigma}
  |v_\epsilon\rangle & = \vartheta_4(t+\epsilon\pi/3,p^2)|{\uparrow}\rangle+\vartheta_1(t+\epsilon\pi/3,p^2)|{\downarrow}\rangle,\\
  \langle w_\epsilon| & = \epsilon\left(-\vartheta_1(t-\epsilon\pi/3,p^2)\langle {\uparrow}|+\vartheta_4(t-\epsilon\pi/3,p^2)\langle{\downarrow}|\right),
    \label{eqn:DefWSigma}
\end{align}
where $\epsilon = \pm$. One checks that
\begin{equation}
  \langle w_\epsilon|v_{\epsilon'}\rangle = \vartheta_1(\pi/3,p)\vartheta_2(t,p)\delta_{\epsilon\epsilon'},
  \label{eqn:ScalarProdWV}
\end{equation}
for each $\epsilon,\epsilon'=\pm$. In the next five lemmas, we establish several properties of these states.

\begin{lemma}
\label{lem:Basis1}
For all $(p,t) \in \bar{\mathbb D}$ with $\,t\neq \pi/2$, the states $|v_+\rangle$ and $|v_-\rangle$ form a basis of $V$.
\begin{proof}
The matrix 
\begin{equation}
  M =\begin{pmatrix}
    \vartheta_4(t+\pi/3,p^2) & \vartheta_4(t-\pi/3,p^2)\\
    \vartheta_1(t+\pi/3,p^2) & \vartheta_1(t-\pi/3,p^2)\\    
  \end{pmatrix},
\end{equation}
whose columns are given by $|v_+\rangle$ and $|v_-\rangle$, has the determinant
\begin{equation}
  \det M = -\vartheta_1(\pi/3,p)\vartheta_2(t,p).
\end{equation}
For $t \neq \pi/2$, this determinant is non-vanishing. Hence the vectors are linearly independent. Therefore, they form a basis of $V$.
\end{proof}
\end{lemma}

If $t=\pi/2$, then $|v_+\rangle$ and $|v_-\rangle$ are not linearly independent: We have $|v_-\rangle = |v_+\rangle$. To find a suitable basis of $V$, we define
\begin{equation}
   \label{eqn:DefVPlusBar}
|\dot v_+\rangle = \left.\frac{\diff}{\diff t}|v_+\rangle\right|_{t=\pi/2}.  
\end{equation}
    
\begin{lemma}
\label{lem:Basis2}
For all $(p,t) \in \bar{\mathbb D}$ with $t=\pi/2$, the states $|v_+\rangle$ and $|\dot v_+\rangle$ form a basis of $V$.
\begin{proof}
  The matrix whose columns are given by the states $|v_+\rangle,|\dot v_+\rangle$ is  \begin{equation}
\dot M =\begin{pmatrix}
    \vartheta_3(\pi/3,p^2) & \vartheta_3'(\pi/3,p^2)\\
    \vartheta_2(\pi/3,p^2) & \vartheta_2'(\pi/3,p^2)   
     \end{pmatrix}.
 \end{equation}
 Its determinant is given by $\det \dot M = -\frac{1}{2}\vartheta_1'(0,p)\vartheta_1(\pi/3,p)$, which is non-zero. Hence the vectors are linearly independent. Therefore, they form a basis of $V$.
\end{proof}
\end{lemma}

\begin{lemma}
  \label{lem:QOnBasis}
  For each $\epsilon=\pm$, we have
  \begin{equation}
    \q|v_\epsilon\rangle=\Lambda_\epsilon|v_\epsilon\rangle\otimes |v_\epsilon\rangle,
    \label{eqn:QOnV}
  \end{equation}
 where
  \begin{equation}
    \label{eqn:DefLambdaSigma}
    \Lambda_\epsilon = \frac{2\epsilon \vartheta_1(\pi/3,p^2)\vartheta_4(0,p^2)^2}{\vartheta_4(\pi/3,p^2)\vartheta_2(0,p)}\frac{\vartheta_2(t+\epsilon\pi/3,p)}{\vartheta_4(t,p^2)^3}.
  \end{equation}
  \begin{proof}
    The proof follows from a number of identities for Jacobi theta functions.
  \end{proof}
\end{lemma}
\begin{lemma}
  \label{lem:QOnVBar}
  Let $t = \pi/2$, then
  \begin{equation}
    \q|\dot v_+\rangle = \dot \Lambda_+|v_+\rangle\otimes |v_+\rangle+\Lambda_+(|\dot v_+\rangle\otimes |v_+\rangle+|v_+\rangle\otimes |\dot v_+\rangle),
  \end{equation}
  where
  \begin{equation}
 \dot \Lambda_+ = \left.\frac{\diff}{\diff t}\Lambda_+\right|_{t=\pi/2}.
  \end{equation}
  \begin{proof} 
We differentiate \eqref{eqn:QOnV} at $t=\pi/2$. We eliminate the terms that involve the derivative of $\q$ with respect to $t$ by observing that for $t=\pi/2$, 
\begin{equation}
\left.\frac{\diff}{\diff t}\Lambda_-\right|_{t=\pi/2} =-\dot \Lambda_+\quad \text{and}\quad \left.\frac{\diff}{\diff t}|v_-\rangle\right|_{t=\pi/2} = -|\dot v_+\rangle.
\end{equation}
This leads to the action of $\q$ on $|\dot v_+\rangle$.
  \end{proof}
\end{lemma}
\begin{lemma}
 \label{lem:DualBasis}
  For each $\epsilon,\epsilon'=\pm$, we have
\begin{equation}
  \left(\langle w_\epsilon|\otimes \langle w_{\epsilon'}|\right)\q=\vartheta_1(\pi/3,p)\vartheta_2(t,p)\Lambda_\epsilon \delta_{\epsilon\epsilon'}\langle w_\epsilon|.
\end{equation}
\begin{proof}
  The proof is a straightforward calculation using standard identities for the Jacobi theta functions.
\end{proof}
\end{lemma}

\subsection{Supersymmetry singlets: (co)homology}

\label{sec:Cohomology}

It follows from \eqref{eqn:HFromQ} that the spectrum of a supersymmetric Hamiltonian $H$ is non-negative. If it contains the eigenvalue $E=0$, then the corresponding eigenstates are the solutions of the equations
\begin{subequations}
\label{eqn:E0StateEqs}
\begin{equation}
  \Q|\Psi\rangle = 0,\quad \text{for } L\geqslant 1,
\end{equation}
and
\begin{equation}
\Q^\dagger|\Psi\rangle = 0, \quad \text{for }L\geqslant 2.
\end{equation}
\end{subequations}
These eigenstates are called \textit{supersymmetry singlets} or zero-energy states. The aim of this and the following section is to investigate the absence or existence of supersymmetry singlets for the Hamiltonian $H$ of \cref{prop:HSUSYHXYZ} as a function of $(p,t)\in \bar{\mathbb D}$. To this end, we exploit the relation between supersymmetry and (co)homology \cite{witten:82}.

\paragraph{(Co)homology.} Let $\Q$ denote a generic supercharge and $\Q^\dagger$ its adjoint. For $L=1$, we define $\mathcal H^1 = \ker \{\Q:V\to V^{2}\}$ and $\mathcal H_1 = V/\text{im}\{\Q^\dagger:V^{2}\to V\}$. For each $L\geqslant 2$, we define the quotient spaces
\begin{align}
  \mathcal H^L = \frac{\ker \{\Q:V^L\to V^{L+1}\}}{\text{im}\{\Q:V^{L-1}\to V^{L}\}} \quad\text{and}\quad \mathcal H_L = \frac{\ker \{\Q^\dagger:V^L\to V^{L-1}\}}{\text{im}\{\Q^\dagger:V^{L+1}\to V^{L}\}}.
\end{align}
The direct sums
\begin{equation}
  \mathcal H^\bullet = \bigoplus_{L=1}^\infty \mathcal H^L \quad\text{and}\quad \mathcal H_\bullet = \bigoplus_{L=1}^\infty \mathcal H_L 
\end{equation}
are often referred to as the cohomology of the supercharge $\Q$ and the homology of the adjoint supercharge $\Q^\dagger$, respectively \cite{masson:08}. The space of the supersymmetry singlets of $H$ is isomorphic to both $\mathcal H^L$ and $\mathcal H_L$ for each $L\geqslant 1$. Hence, the computation of the (co)homology allows us to investigate the absence or existence of supersymmetry singlets. 

Let us briefly recall some terminology and notation \cite{hagendorf:17,hagendorf:18}. For $L\geqslant 2$, the elements of $\mathcal H^L$ are equivalence classes of states that are annihilated by the supercharge $\Q$. These states are called representatives. We write $[|\phi\rangle]\in \mathcal H^L$ for the equivalence class of a representative $|\phi\rangle\in  \text{ker}\{\Q:V^L\to V^{L+1}\}$. For $L\geqslant 1$, the elements of $\mathcal H_L$ are equivalence classes, too. If $L=1$, then they are represented by states $|\phi'\rangle \in V$; if $L\geqslant 2$, they are represented by states $|\phi'\rangle \in  \text{ker}\{\Q^\dagger:V^L\to V^{L-1}\}$. As before, we denote the equivalence class of such a representative by $[|\phi'\rangle]$.

\paragraph{Auxiliary results.} To compute the (co)homology for the supercharge of the XYZ Hamiltonian, we establish three auxiliary results.

\begin{lemma}
 \label{lem:Lem1}
 Let $|u_+\rangle,\, |u_-\rangle$ be a basis of $V$ and $\q$ a local supercharge defined by
 \begin{equation}
   \q|u_+\rangle = |u_+\rangle \otimes |u_+\rangle, \quad \q|u_-\rangle = |u_-\rangle \otimes |u_-\rangle,
 \end{equation}
  then $\mathcal H^L = 0$ for each $L\geqslant 1$.
 \begin{proof}
   For $L=1$ the statement $\mathcal H^1=0$ is immediate since $|u_+\rangle$ and $|u_-\rangle$ form a basis of $V$.
    
For $L\geqslant 2$, let $|\psi\rangle\in \ker \{\Q:V^L\to V^{L+1}\}$. We write $|\psi\rangle =  |u_+\rangle\otimes |\psi_+\rangle  + |u_-\rangle\otimes|\psi_-\rangle$ with unique states $|\psi_+\rangle,\,|\psi_-\rangle \in V^{L-1}$. It follows from $\Q\ket{\psi} =0$ that
   \begin{equation}
     |u_+\rangle\otimes(|u_+\rangle \otimes |\psi_+\rangle+\Q|\psi_+\rangle)+|u_-\rangle\otimes(|u_-\rangle \otimes |\psi_-\rangle+\Q|\psi_-\rangle)=0.
   \end{equation}
   Since $|u_+\rangle$ and $|u_-\rangle$ form a basis of $V$, we find $|u_\pm\rangle \otimes |\psi_\pm\rangle = -\Q|\psi_\pm\rangle$. Therefore, we have 
   \begin{equation}
    |\psi\rangle = -\Q(|\psi_+\rangle + |\psi_-\rangle) \in \text{im}\{\Q:V^{L-1} \to V^L\}.
   \end{equation}
   This implies that $\mathcal H^L = 0$.
    \end{proof}
\end{lemma}

\begin{lemma}
  \label{lem:Lem2}
  Let $|u_+\rangle,\, |u_-\rangle$ be a basis of $V$ and $\q$ a local supercharge defined by
  \begin{equation}
    \q|u_+\rangle = |u_+\rangle \otimes |u_+\rangle, \quad \q|u_-\rangle = |u_+\rangle\otimes |u_-\rangle + |u_-\rangle \otimes |u_+\rangle,
  \end{equation}
  then $\mathcal H^L = 0$ for each $L\geqslant 1$.
  \begin{proof}
    For $L=1$, $\mathcal H^1=0$ follows immediately from the fact that $|u_+\rangle,\, |u_-\rangle$ is a basis of $V$. 
    
   For $L\geqslant 2$, let $|\psi\rangle\in \ker \{\Q:V^L\to V^{L+1}\}$. Again, we write $|\psi\rangle =  |u_+\rangle\otimes |\psi_+\rangle  + |u_-\rangle\otimes|\psi_-\rangle$ with unique states $|\psi_+\rangle,|\psi_-\rangle \in V^{L-1}$. The condition $\Q\ket\psi =0$ yields
  \begin{equation}
|u_+\rangle\otimes(|\psi\rangle +\Q |\psi_+\rangle)+ |u_-\rangle \otimes(|u_+\rangle \otimes |\psi_-\rangle + \Q|\psi_-\rangle) = 0.
  \end{equation}
  Since $|u_+\rangle$ and $|u_-\rangle$ span $V$, we obtain
  \begin{equation}
    |\psi\rangle = -\Q|\psi_+\rangle \in \text{im}\{\Q:V^{L-1} \to V^L\}.
   \end{equation}
   Hence, $\mathcal H^L = 0$.
  \end{proof}
\end{lemma}

\begin{lemma}
\label{lem:Lem3}
  Let $|u_+\rangle,\, |u_-\rangle$ be a basis of $V$ and $\q$ a local supercharge defined by
  \begin{equation}
    \q|u_+\rangle=0, \quad \q|u_-\rangle = |u_-\rangle \otimes |u_-\rangle,
  \end{equation}
  then $\mathcal H^L = \mathbb C[|u_+\rangle^{\otimes L}]$ for each $L\geqslant 1$.
  \begin{proof}
  For each $L\geqslant 1$, we define a mapping $S:V^{L}\to V^{L+1}$ by
  \begin{equation}
    S|\psi\rangle =  |u_+\rangle\otimes |\psi\rangle .
  \end{equation}
   It satisfies the commutation relation $S\Q =-\Q S$ on $V^L$. Hence, the mapping $S^\sharp:\mathcal H^{L} \to \mathcal H^{L+1}$, given by
  \begin{equation}
    S^\sharp[|\psi\rangle] = [|u_+\rangle\otimes|\psi\rangle ], 
  \end{equation}
  is well defined \cite{masson:08}. We prove that $S^\sharp$ is a bijection.
    
  First, we show that $S^\sharp$ is injective. This is straightforward for $L=1$. For $L\geqslant 2$, we show that the kernel of $S^\sharp$ is zero in the cohomology. This is equivalent to the statement that any state $|\psi\rangle \in \text{ker}\{\Q:V^L\to V^{L+1}\}$ with
   \begin{equation}
      S|\psi\rangle = \Q|\phi\rangle,
    \end{equation}
   for some $|\phi\rangle \in V^L$, belongs to $\text{im}\,\{\Q:V^{L-1}\to V^L\}$.
   To see this, we write $|\phi\rangle = |u_+\rangle\otimes |\phi_+\rangle + |u_-\rangle\otimes |\phi_-\rangle$ with unique states $|\phi_+\rangle,\,|\phi_-\rangle\in V^{L-1}$. It follows that
   \begin{equation}
   \ket{u_+}\otimes  |\psi\rangle = -\ket{u_+}\otimes \Q|\phi_+\rangle -\ket{u_-}\otimes \left( \ket{u_-}\otimes \ket{\phi_-} +\Q|\phi_-\rangle \right).
   \end{equation}
   Since $|u_+\rangle,\, |u_-\rangle$ form a basis of $V$, we infer $|\psi\rangle =- \Q|\phi_+\rangle$, which proves the injectivity.
   
   Second, we show that $S^\sharp$ is surjective. To this end, we fix $L\geqslant 2$ and  consider a representative 
   $|\psi\rangle \in V^{L}$ of an element of $\mathcal H^{L}$. As before, we write 
   $|\psi\rangle = \ket{u_+}\otimes |\psi_+\rangle +\ket{u_-}\otimes |\psi_-\rangle$ with unique states $|\psi_+\rangle,|\psi_-\rangle \in V^{L-1}$. The equation $\Q|\psi\rangle = 0$ implies
   \begin{equation}
     \Q|\psi_+\rangle = 0, \quad \Q |\psi_-\rangle = -\ket{u_-}\otimes \ket{\psi_-},
   \end{equation}
   and therefore
   \begin{equation}
    |\psi\rangle = |u_+\rangle \otimes |\psi_+\rangle - \Q |\psi_-\rangle.
   \end{equation}
   Hence, $[|\psi\rangle] = [|u_+\rangle\otimes |\psi_+\rangle ] = S^\sharp[|\psi_+\rangle]$ with $|\psi_+\rangle\in \text{ker}\{\Q:V^{L-1}\to V^{L}\}$. This proves the surjectivity.
   
   Since $S^\sharp$ is a bijection, it follows that $\mathcal H^{L} = (S^\sharp)^{L-1}(\mathcal H^{1})$ for each $L\geqslant 2$. One checks that 
$\mathcal H^1 = \mathbb C[|u_+\rangle]$. Hence, $\mathcal H^{L} =\mathbb C[|u_+\rangle^{\otimes L}]$.
  \end{proof}
\end{lemma}

\paragraph{Results for the XYZ supercharge.} In the remainder of this section, $\Q$ denotes the supercharge constructed from the local supercharge \eqref{eqn:XYZLocalSupercharge} for the XYZ Hamiltonian. We apply the auxiliary results to this case.

\begin{proposition}
  \label{prop:Cohom}
  Let $L\geqslant 1$, and $(p,t)\in \bar{ \mathbb D}$.
  We have
  \begin{equation}
    \mathcal H^L =
    \begin{cases}
      0, & \text{if } t\neq \pi/6,\\
      \mathbb C[|v_+\rangle^{\otimes L}], &\text{if } t=\pi/6.
    \end{cases}
  \end{equation}
  \begin{proof}
    We distinguish three cases.
    
    First, we consider $t \neq \pi/2,\pi/6$. In this case, it follows from \cref{lem:Basis1} that $|v_+\rangle$ and $|v_-\rangle$ form a basis of $V$. Furthermore, the constants $\Lambda_\pm$, defined in \eqref{eqn:DefLambdaSigma}, are non-vanishing. Hence, the states
    \begin{equation}
      |u_+\rangle = \Lambda_+ |v_+\rangle, \quad |u_-\rangle = \Lambda_- |v_-\rangle
\end{equation}
form a basis of $V$. We find from \cref{lem:QOnBasis} that $\q|u_+\rangle= |u_+\rangle \otimes |u_+\rangle$, $\q|u_-\rangle= |u_-\rangle \otimes |u_-\rangle$. Hence, we apply \cref{lem:Lem1} and conclude that $\mathcal H^L=0$ for each $L\geqslant 1$.
    
 Second, we suppose that $t = \pi/2$. It follows from \cref{lem:Basis2} that the states $|v_+\rangle$ and $|\dot v_+\rangle$, defined in \eqref{eqn:DefVPlusBar}, form a basis of $V$. We define the states
    \begin{equation}
       |u_+\rangle = \Lambda_+ |v_+\rangle, \quad
       |u_-\rangle = \dot \Lambda_+|v_+\rangle+\Lambda_+ |\dot v_+\rangle.
       \end{equation}
   These states form a basis of $V$ because $\Lambda_+,\dot \Lambda_+\neq 0$ for $t=\pi/2$. Moreover, we have $\q|u_+\rangle = |u_+\rangle \otimes |u_+\rangle, \, \q|u_-\rangle = |u_+\rangle \otimes |u_-\rangle + |u_-\rangle \otimes |u_+\rangle$, thanks to \cref{lem:QOnVBar}. Therefore, it follows from \cref{lem:Lem2} that $\mathcal H^L =0$ for each $L\geqslant 1$.
    
Third, we analyse the case where $t=\pi/6$. In this case, we have $\Lambda_+=0$ and $\Lambda_-\neq 0$. The states
\begin{equation}
  |u_+\rangle = |v_+\rangle, \quad |u_-\rangle = \Lambda_-|v_-\rangle 
\end{equation}  
 constitute a basis of $V$. They obey the relations $\q |u_+\rangle = 0$ and $\q|u_-\rangle = |u_-\rangle \otimes |u_-\rangle$. According to \cref{lem:Lem3}, we have
    \begin{equation}
      \mathcal H^L = \mathbb C[|u_+\rangle^{\otimes L}] =\mathbb C[|v_+\rangle^{\otimes L}],
    \end{equation}
    for each $L\geqslant 1$.
     \end{proof}
\end{proposition}

\begin{proposition}
\label{prop:Hom}
Let $L\geqslant 1$ and $(p,t) \in \mathbb{\bar D}$. We have
  \begin{equation}
    \mathcal H_L =
    \begin{cases}
      0, &\text{if } t\neq \pi/6,\\
      \mathbb C[|w_+\rangle^{\otimes L}], &\text{if } t=\pi/6.
    \end{cases}
  \end{equation}
\begin{proof}
  First, we consider $t\neq \pi/6$. In this case, $\mathcal H_L=0$ for each $L\geqslant 1$ follows immediately from \cref{prop:Cohom} and the fact that $\mathcal H^L$ and $\mathcal H_L$ are isomorphic.
  
  Second, we consider $t = \pi/6$ and compute $\mathcal H_L$. To this end, we note that \cref{lem:DualBasis} implies
  \begin{equation}
    \label{eqn:QDaggerPi6}   
 \Q^\dagger\left(|w_\epsilon\rangle \otimes |w_{\epsilon'}\rangle\right) = -\vartheta_1(\pi/3,p)^2\Lambda_\epsilon\delta_{\epsilon\epsilon'}|w_\epsilon\rangle,
  \end{equation}
  for each $\epsilon,\epsilon'=\pm$.
 Furthermore, we have $\Lambda_+=0$ and $\Lambda_- \neq 0$. For $L=1$, we find
  \begin{equation}
    \mathcal H_1 = V/\text{im}\{\Q^{\dagger}:V^2\to V\}=V/\mathbb C|w_-\rangle = \mathbb C [|w_+\rangle].
  \end{equation}
  For $L\geqslant 2$, \cref{prop:Cohom} implies that $\mathcal H_L$ is one-dimensional. Hence, $\mathcal H_L = \mathbb C[|\omega\rangle]$ for some $|\omega\rangle\in V^L$ that is in the kernel of $\Q^\dagger$, but not in its image. We claim that $|\omega\rangle = |w_+\rangle^{\otimes L}$ is a valid choice. Indeed, on the one hand \eqref{eqn:QDaggerPi6} implies $\Q^\dagger|\omega\rangle = 0$. On the other hand, we use \eqref{eqn:ScalarProdWV} to compute the scalar product
  \begin{equation}
    \langle \omega|\left(|v_+\rangle^{\otimes L}\right) = \langle w_+|v_+\rangle^L  =\vartheta_1(\pi/3,p)^{2L}, 
  \end{equation}
  which is non-zero. If $|\omega\rangle = \Q^\dagger |\phi\rangle$ for some $|\phi\rangle \in V^{L+1}$ then 
$\langle \omega|\left(|v_+\rangle^{\otimes L}\right)= \langle \phi|\Q\left(|v_+\rangle^{\otimes L}\right) = 0$. This is a contradiction and therefore proves the claim.
  \end{proof}
\end{proposition}

\subsection{Supersymmetry singlets: spin-chain ground states}

\label{sec:Singlets}

\paragraph{(Co)homology decompositions.} Let $\Q$ be a generic supercharge and $\Q^\dagger$ its adjoint. We recall the relations between their (co)homology and the supersymmetry singlets of the corresponding Hamiltonian $H$ \cite{witten:82}.

For $L=1$, any $|\Psi\rangle\in \mathcal H^1$ trivially is a singlet. For $L\geqslant 2$, let $|\phi\rangle$ represent a non-zero element of $\mathcal H^L$, then there is a state $|\gamma\rangle \in V^{L-1}$ such that 
\begin{equation}
  \label{eqn:CohomDecomp}
  |\Psi\rangle = |\phi\rangle + \Q|\gamma\rangle
\end{equation}
is a supersymmetry singlet. Conversely, any supersymmetry singlet can be written as a sum of a representative of a non-zero element of $\mathcal H^L$ and a state in the image of the supercharge. 

Likewise, let $L\geqslant 1$ and $|\phi'\rangle$ represent a non-zero element of $\mathcal H_L$ then there is a state $|\gamma'\rangle \in V^{L+1}$ such that 
\begin{equation}
  \label{eqn:HomDecomp}
  |\Psi\rangle = |\phi'\rangle + \Q^\dagger|\gamma'\rangle
\end{equation}
is a supersymmetry singlet. Conversely, any supersymmetry singlet can be written as a sum of a representative of a non-zero element of $\mathcal H_L$ and a state in the image of the adjoint supercharge.

In the following, we refer to \eqref{eqn:CohomDecomp} and \eqref{eqn:HomDecomp} as a cohomology and homology decomposition of a supersymmetry singlet $|\Psi\rangle$, respectively. For the XYZ supercharge and $t=\pi/6$, we use these decompositions to characterise the space of ground states of the Hamiltonian $H$.

\begin{theorem}
  \label{thm:CohomRep}
  Let $L\geqslant 1$ and $(p,t)\in \bar{\mathbb D}$. If $t\neq \pi/6$, then the Hamiltonian $H$ does not possess supersymmetry singlets. Conversely, if $t=\pi/6$ then the space of supersymmetry singlets of $H$ is one-dimensional, and spanned by
  \begin{equation}
    \label{eqn:CohomologyRepresentation}
    |\Psi_L\rangle =
    \begin{cases}
      |v_+\rangle, & L=1,\\
      |v_+\rangle^{\otimes L} + \Q|\gamma_L\rangle,& L\geqslant 2,
    \end{cases}
  \end{equation}
 where  $|\gamma_L\rangle \in V^{L-1}$.
 \begin{proof}
    First, we consider $t\neq \pi/6$. In this case, it follows from \cref{prop:Cohom} that $\mathcal H^L=0$. Hence, $H$ does not possess supersymmetry singlets.
    
    Second, for $t=\pi/6$ the \cref{prop:Cohom} states that $\mathcal H^L=\mathbb C[|v_+\rangle^{\otimes L}]$. Hence, the space of the supersymmetry singlets of $H$ is one-dimensional. In fact, the decomposition for $L\geqslant 2$ follows from \eqref{eqn:CohomDecomp}.
  \end{proof}
\end{theorem}

\begin{proposition}
  For $t=\pi/6$ and each $L\geqslant 1$, the state \eqref{eqn:CohomologyRepresentation} can be written as
  \begin{equation}
    \label{eqn:HomologyRepresentation}
    |\Psi_L\rangle = \mu_L|w_+\rangle^{\otimes L} + \Q^\dagger|\gamma'_L\rangle,
  \end{equation}
  with $|\gamma'_L\rangle \in V^{L+1}$. The constant $\mu_L$ is non-zero and given by
  \begin{equation}
  \label{eqn:DefMuL}
    \mu_L = \frac{\left(\langle v_+ |^{\otimes L}\right)|\Psi_L\rangle}{\vartheta_1(\pi/3,p)^{2L}}.
  \end{equation}
  \begin{proof}
    The decomposition \eqref{eqn:HomologyRepresentation} follows from $\mathcal H_L = \mathbb C[|w_+\rangle^{\otimes L}]$ for $t=\pi/6$, found in \cref{prop:Hom}. To find the coefficient $\mu_L$, it is sufficient to compute the scalar product of both sides of \eqref{eqn:HomologyRepresentation} with $|v_+\rangle^{\otimes L}$. It has to be non-zero, because otherwise $|\Psi_L\rangle$ would be in the image of $\Q^\dagger$. This would imply $|\Psi_L\rangle=0$ \cite{hagendorf:17} and thus contradict \cref{prop:Hom}.
  \end{proof}
\end{proposition}

\paragraph{Alternative decompositions.} The (co)homology decompositions \eqref{eqn:CohomDecomp} and \eqref{eqn:HomDecomp} of a supersymmetry singlet $|\psi\rangle$ are not unique. The reason is that the representatives of $|\phi\rangle$ and $|\phi'\rangle$ are only defined up to a state in the image of $\Q$ or $\Q^\dagger$, respectively. We exploit the non-uniqueness to compute two alternative decompositions for the supersymmetry singlet $|\Psi_L\rangle$. To this end, we define
\begin{align}
  \begin{split}
  |\chi\rangle &= \ket{v_+}\otimes \ket{v_+}- \kappa^2\ket{v_-}\otimes \ket{v_-},\\
|\alpha\rangle &= \ket{w_+}\otimes \ket{w_+} +\kappa^{-1} \ket{w_-}\otimes \ket{w_+},
\end{split}
\end{align}
where $ \kappa = \vartheta_3(\pi/3,p)/\vartheta_3(0,p)$.
\begin{proposition}
For $t=\pi/6$ and each $L\geqslant 2$, the supersymmetry singlet $|\Psi_L\rangle$ can be written as
\begin{equation}
  \label{eqn:AlternativeDecomp1}
  |\Psi_L\rangle = |\chi\rangle \otimes |v_+\rangle^{\otimes (L-2)}+ \Q|\delta_L\rangle,
\end{equation}
for some state $|\delta_L\rangle \in V^{L-1}$, and as
\begin{equation}
  \label{eqn:AlternativeDecomp2}
  |\Psi_L\rangle = \mu_L |\alpha\rangle \otimes |w_+\rangle^{\otimes (L-2)} + \Q^\dagger|\delta_L'\rangle,
\end{equation}
for some state $|\delta_L'\rangle \in V^{L+1}$. Here, $\mu_L$ is the constant defined in \eqref{eqn:DefMuL}.
\begin{proof}
  The proof consists of two simple calculations. We focus on \eqref{eqn:AlternativeDecomp1}. Using $\q|v_+\rangle=0,\,\q|v_-\rangle=\Lambda_-|v_-\rangle \otimes |v_-\rangle$ with $\Lambda_-\neq 0$ for $t=\pi/6$, we obtain
  \begin{equation}
    |v_+\rangle^{\otimes L} =|\chi\rangle \otimes |v_+\rangle^{\otimes (L-2)}-\Q\left(\kappa^2\Lambda_-^{-1}|v_-\rangle\otimes |v_+\rangle^{\otimes (L-2)}\right).
  \end{equation}
  We use this in \eqref{eqn:CohomologyRepresentation} and obtain \eqref{eqn:AlternativeDecomp1} with
  \begin{equation}
    |\delta_L\rangle = |\gamma_{L}\rangle-\kappa^2\Lambda_-^{-1}|v_-\rangle\otimes |v_+\rangle^{\otimes (L-2)}.
  \end{equation}
  The proof of \eqref{eqn:AlternativeDecomp2} is similar.
\end{proof}
\end{proposition}

Finally, we point out that for $t=\pi/6$, the basis states $|v_\pm\rangle$ and their duals $|w_\pm\rangle$, as well as $|\chi\rangle$ and $|\alpha\rangle$, can up to factor be written in terms of polynomials in $\zeta$ and $y$. This property can be shown with the help of identities between Jacobi theta functions.
\begin{lemma}
  \label{lem:VWPolynomial}
  We have $|v_\pm\rangle = C_{\pm}|\bar v_\pm\rangle$ and $|w_\pm\rangle = C_\mp|\bar w_\pm\rangle$, where
  \begin{align}
    |\bar v_+\rangle &= y(1-\zeta y^2)|{\uparrow}\rangle + (\zeta-y^2)|{\downarrow}\rangle, & |\bar v_-\rangle &= |{\uparrow}\rangle - y|{\downarrow}\rangle,\\
    |\bar w_+\rangle &= y|{\uparrow}\rangle + |{\downarrow}\rangle, & |\bar w_-\rangle &= (\zeta-y^2)|{\uparrow}\rangle - y(1-\zeta y^2)|{\downarrow}\rangle,
  \end{align}
  and $C_+ = (1-\zeta^2)^{-2/3}y^{-1}\vartheta_3(\pi/3,p^2),\,C_- = \vartheta_3(\pi/3,p^2)$.
\end{lemma}
\begin{lemma}
  \label{lem:ChiAlphaPolynomial}
 We have $|\chi\rangle = D_+|\bar \chi\rangle$ and $|\alpha\rangle = D_-|\bar \alpha\rangle$ with
  \begin{align}
    |\bar \chi\rangle &= y^2(\zeta-2+\zeta y^2)|{\uparrow\uparrow}\rangle+ y (y^2-1)(|{\uparrow\downarrow}\rangle + |{\downarrow\uparrow}\rangle) - (\zeta+(\zeta-2)y^2)|{\downarrow\downarrow}\rangle,\\
    |\bar \alpha\rangle &=  y\left(|{\uparrow\uparrow}\rangle-|{\downarrow\downarrow}\rangle\right) +|{\uparrow\downarrow}\rangle-y^2|{\downarrow\uparrow} \rangle,
  \end{align}
  where $D_+=\zeta(y^2-1)C_+^2 $ and $D_-=\zeta (y^2-1)y^{-1}(\zeta-1)^{-1}C_-^2$.
\end{lemma}

\paragraph{The XYZ ground states.} We now return to the XYZ Hamiltonian defined in \eqref{eqn:SUSYXYZHamiltonian} and prove \cref{thm:MainTheorem1}. To this end, we introduce the polynomial
\begin{equation}
  P(\zeta,y) = \zeta(1+y^4)-(3-\zeta^2)y^2.
\end{equation}
It is straightforward to see that, given $0<\zeta<1$, the biquadratic equation $P(\zeta,y)=0$ for $y$ possesses four real solutions. They have particularly simple expressions in the parameterisation by Jacobi theta functions.
\begin{lemma}
\label{lem:ZeroesP}
Let $0<\zeta<1$ and $y$ be parametrised according to \eqref{eqn:JacobiThetaParam} with $0<p<1$, then the solutions of $P(\zeta,y)=0$ are given by
\begin{equation}
  y_0 = \frac{\vartheta_1(\pi/6,p^2)}{\vartheta_4(\pi/6,p^2)}, \quad y_1 = \frac{\vartheta_4(\pi/6,p^2)}{\vartheta_1(\pi/6,p^2)}, \quad y_2 =- \frac{\vartheta_4(\pi/6,p^2)}{\vartheta_1(\pi/6,p^2)}, \quad  y_3=-\frac{\vartheta_1(\pi/6,p^2)}{\vartheta_4(\pi/6,p^2)}.
  \label{eqn:ZeroesP}
\end{equation}
In particular, $P(\zeta,y)=0$ for $(\zeta,y)\in \mathbb D$ if and only if $y=y_0$.
\begin{proof}
  First, we substitute the parameterisation \eqref{eqn:JacobiThetaParam} into the polynomial $P(\zeta,y)$ and find
  \begin{equation}
  P(\zeta,y) = \frac{C\vartheta_1(\pi/6-t,p)\vartheta_1(\pi/6+t,p)}{\vartheta_4(t,p^2)^4},
\end{equation}
where $C=(\vartheta_1(\pi/3,p)\vartheta_4(0,p^2)/\vartheta_4(\pi/3,p^2))^2$. The right-hand side vanishes if and only if
\begin{equation}
  t = \pm \pi/6,\, t= \pm \pi/6 +\i  s \mod \pi,2\i s,
\end{equation}
where $s>0$ is defined through $p=e^{-s}$. The evaluations of $y$ at these values of $t$ lead to the four roots given in \eqref{eqn:ZeroesP}.

Second, we conclude from \eqref{eqn:ZeroesP} that $(\zeta,y_0)\in \mathbb D$ but $(\zeta,y_\alpha)\notin \mathbb D$ for $\alpha=1,2,3$.
\end{proof}
\end{lemma}

In terms of the parameterisation \eqref{eqn:JacobiThetaParam}, this lemma implies that $P(\zeta,y)$ vanishes for $(p,t)\in \bar{\mathbb D}$ if and only if $t=\pi/6$. We exploit this property in the following proof.

\begin{proof}[Proof of \cref{thm:MainTheorem1}]

First, we prove the theorem for $(\zeta,y)\in \mathbb D$. To this end, we recall the relation \eqref{eqn:RelationsHam} that expresses the Hamiltonian $H$ in terms of $H_{\textup{\tiny XYZ}}$ for $L\geqslant 1$ sites:
\begin{equation}
  \label{eqn:RelationHams}
  H = x\left(H_{\textup{\tiny XYZ}}+\frac{(L-1)(3+\zeta^2)}{4}+2\lambda_0\right).
\end{equation}
The factor $x$ in this relation is positive for all $(\zeta,y)\in \mathbb D$. Hence, the spaces of the ground states of $H$ and $H_{\textup{\tiny XYZ}}$ are equal. We use the parameterisation of $(\zeta,y)\in \mathbb D$ by $(p,t)\in \bar{\mathbb D}$. According to \cref{thm:CohomRep}, the space of the ground states of $H$ is spanned by the supersymmetry singlet $|\Psi_L\rangle$ if and only if $t=\pi/6$. We use \cref{lem:ZeroesP} to conclude that the space of the ground states of $H_{\textup{\tiny XYZ}}$ consists of supersymmetry singlets if and only if $y=y_0$. According to \eqref{eqn:RelationsHam} the corresponding ground-state eigenvalue of this Hamiltonian is
\begin{equation}
    E_0 = \left.-\frac{(L-1)(\zeta^2+3)}{4}-2\lambda_0\right|_{y=y_0}=-\frac{(L-1)(\zeta^2+3)}{4}-\frac{(1+\zeta)^2}{2}.
  \end{equation}
  
  Second, we consider $0<\zeta<1$ and $(\zeta,y)\notin \mathbb D$. In this case, it follows from \eqref{eqn:Pi} that there is an integer $1\leqslant \alpha \leqslant 3$ such that
  \begin{equation}
H_{\textup{\tiny XYZ}}(\zeta,y) = \mathcal R^\alpha(-\pi)H_{\textup{\tiny XYZ}}(\zeta,\bar y)\mathcal R^\alpha(\pi)
  \end{equation}
  with $(\zeta,\bar y)\in \mathbb D$. Since $\mathcal R^\alpha(\pi)$ is a unitary operator, the two Hamiltonians in this equality have the same spectrum. Furthermore, writing $\Q=\Q(\zeta,y)$ to indicate the dependence of the supercharge on $\zeta$ and $y$, we have 
\begin{equation}
\Q (\zeta,y_\alpha)=\mathcal R^\alpha (-\pi) \Q (\zeta,y_0) \mathcal R^\alpha(\pi).
\end{equation}
The state $|\Psi^{\alpha}_L\rangle = \mathcal R^\alpha(-\pi)|\Psi_L\rangle$ is a supersymmetry singlet with respect to the supercharge $\Q(\zeta, y_\alpha)$.
We conclude from these two observations that the space of the ground states of $H_{\textup{\tiny XYZ}}(\zeta,y)$ is a space of supersymmetry singlets if and only if $\bar y = y_0$, and hence $y = y_\alpha$. This space is one-dimensional and spanned by the supersymmetry singlet $|\Psi^\alpha_L\rangle$.
\end{proof}

\section{The transfer-matrix eigenvalue}
\label{sec:8V}

The purpose of this section is to prove \cref{thm:MainTheorem2}. To this end, we recall a few elementary properties of the transfer matrix and its relation to the XYZ Hamiltonian in \cref{sec:TM8V}. In \cref{sec:TMSUSY}, we establish a commutation relation between the transfer matrix and the supercharge of the XYZ spin chain. We use this commutation relation in \cref{sec:TMEV} to evaluate the action of the transfer matrix on the supersymmetry singlet $|\Psi_L\rangle$. It allows us to establish the explicit formula for the eigenvalue $\Lambda_L$ and prove the theorem.

\subsection{The transfer matrix}
\label{sec:TM8V}

\paragraph{Transfer matrix.} The transfer matrices of the eight-vertex model on the square lattice can be constructed from its $R$-matrix. This $R$-matrix is an operator $R:V\otimes V \to V\otimes V$. In the canonical basis $|{\uparrow\uparrow}\rangle,|{\uparrow\downarrow}\rangle,|{\downarrow\uparrow}\rangle,|{\downarrow\downarrow}\rangle$ of $V\otimes V$ it reads
\begin{equation}
  R =
  \begin{pmatrix}
    a & 0 & 0 & d\\
    0 & b & c & 0\\
    0 & c & b & 0\\
    d & 0 & 0 & a
  \end{pmatrix},
\end{equation}
where $a,b,c,d$ are the vertex weights.
Let us consider the space $V_0 \otimes V^L = V_0 \otimes V_1 \otimes \cdots \otimes V_L$, where $V_0 = V$ is the so-called auxiliary space. We denote by $R_{ij}, \,0\leqslant i < j \leqslant L,$ the $R$-matrix acting non-trivially only on the factors $V_i$ and $V_j$ of the tensor product $V_0 \otimes V^L$. For convenience, we introduce the abbreviations 
\begin{equation}
  U_{0,[i,j]} =  R_{0j}R_{0j-1}\cdots R_{0i},\quad \bar U_{0,[i,j]} =  R_{0i}R_{0i+1}\cdots R_{0j},
  \label{eqn:DefU}
\end{equation}
for 
$1\leqslant i\leqslant j\leqslant L$. We also define $U_{0,[j+1,j]} = \bar U_{0,[j+1,j]}=1$ for $j=0,\dots,L$.

The transfer matrix of the eight-vertex model for a strip with $L$ vertical lines and open boundary conditions is an operator $\mathcal T:V^L\to V^L$ defined as 
\begin{equation}
  \mathcal T = \text{tr}_0\left(K_0^+ U_{0,[1,L]}K_0^-\bar U_{0,[1,L]}\right).
\end{equation}
The trace is taken with respect to the auxiliary space $V_0$. Moreover, $K_0^\pm$ are operators $K^\pm:V\to V$ acting on the auxiliary space. They are called $K$-matrices and encode the boundary conditions.

To investigate the properties of the transfer matrix, it is often convenient to use a parameterisation of the vertex weights in terms of Jacobi theta functions \cite{baxterbook}. We use
\begin{align}
\label{eqn:param8V}
\begin{split}
a(u)= \rho\,\vth{4}{2\eta}{p^2} \vth{4}{u}{p^2} \vth{1}{u+2\eta}{p^2},\\
b(u)= \rho\,\vth{4}{2\eta}{p^2} \vth{1}{u}{p^2} \vth{4}{u+2\eta}{p^2},\\
c(u)= \rho\,\vth{1}{2\eta}{p^2} \vth{4}{u}{p^2} \vth{4}{u+2\eta}{p^2},\\
d(u)= \rho\, \vth{1}{2\eta}{p^2} \vth{1}{u}{p^2} \vth{1}{u+2\eta}{p^2}.
\end{split}
\end{align}
Here, $\rho$ is a constant, $u$ the \textit{spectral parameter} and $\eta$ the \textit{crossing parameter}. With this parameterisation, the $R$-matrix of the eight-vertex model $R=R(u)$ obeys the Yang-Baxter equation: For all $u,v$, we have
\begin{equation}
  \label{eqn:YBE}
  R_{12}(u-v)R_{13}(u)R_{23}(v)= R_{23}(v)R_{13}(u)R_{12}(u-v).
\end{equation}
Furthermore, we choose
\begin{equation}
  K^- = K(u), \quad K^{+}=K(u+2\eta),
  \label{eqn:ChoiceForKpm}
\end{equation}
where the operator $K=K(u)$ is a solution of the \textit{reflection equation}: For all $u$ and $v$ it obeys
\begin{equation}
  \label{eqn:bYBE}
  R_{12}(u-v)K_1(u)R_{12}(u+v)K_2(v)=K_2(v)R_{12}(u+v)K_1(u)R_{12}(u-v),
\end{equation}
where $K_i(u)$ denotes the operator $K(u)$ acting on $V_i$.

Let us write $\mathcal T = \mathcal T(u)$ to stress the dependence of the transfer matrix on the spectral parameter. The choice \eqref{eqn:ChoiceForKpm} implies that transfer matrices with different spectral parameters commute: We have
\begin{equation}
  \mathcal T(u)\mathcal T(v) = \mathcal T(v)\mathcal T(u),
  \label{eqn:CommT}
\end{equation}
for all $u$ and $v$ \cite{sklyanin:88}. The proof of this commutation relation is based on the Yang-Baxter equation \eqref{eqn:YBE} and the reflection equation \eqref{eqn:bYBE}.

\paragraph{Transfer matrix and Hamiltonian.} We now recall the relation between the transfer matrix and the Hamiltonian of the XYZ spin chain \cite{sklyanin:88}. To this end, we use the $K$-matrix
\begin{equation}
  \label{eqn:KMatrixGeneral}
   K(u) = 1 + \sum_{\alpha=1}^3 \frac{\vartheta_1(u,p)}{\vartheta_{5-\alpha}(u,p)}\mu_\alpha\sigma^\alpha.
 \end{equation}
Here, $\mu_1,\mu_2,\mu_3$ are arbitrary complex numbers. Up to an overall factor, this $K$-matrix is the most general solution to the reflection equation \eqref{eqn:bYBE} of the eight-vertex model \cite{inami:94,hou:95,vega:94}.

\begin{proposition}
  \label{prop:TMLogDer}
 We have the logarithmic derivative
  \begin{equation}
  \label{eqn:TMLogDer}
  \mathcal T(0)^{-1}\mathcal T'(0) = L\left(\frac{a'(0)+c'(0)}{a(0)}\right) - \frac{2b'(0)}{a(0)}H_{\textup{\tiny XYZ}}.
  \end{equation}
  Here, $H_{\textup{\tiny XYZ}}$ is the Hamiltonian \eqref{eqn:XYZHamiltonian} of the open XYZ spin chain with the anisotropy parameters
  \begin{equation}
  \label{eqn:GeneralAP}
  J_1 = 1 + \frac{d'(0)}{b'(0)}, \quad J_2 = 1 - \frac{d'(0)}{b'(0)}, \quad J_3 = \frac{a'(0) - c'(0)}{b'(0)},
\end{equation}
and the boundary terms
\begin{equation}
    \label{eqn:hB}
  h_{\textup{\tiny B}}^\pm = -\frac{\vartheta_1(2\eta,p)}{2} \sum_{\alpha=1}^3 \frac{J_\alpha \mu_\alpha}{\vartheta_{5-\alpha} (2\eta, p)} \sigma^\alpha.
\end{equation}
\begin{proof}
  We have $R(0)=a(0)P$, where $P$ is the permutation operator on $V\otimes V$, $\text{tr}\,K(u)=2,\,\text{tr}\,K'(u)=0$ and $K(u=0)=1$. After a standard calculation, we obtain the logarithmic derivative
  \begin{equation}
    \label{eqn:TMLogDerIntermediate}
    \mathcal T(0)^{-1}\mathcal T'(0) = \frac{2}{a(0)}\sum_{j=1}^{L-1}\check R_{jj+1}'(0) + K_1'(0) + \frac{1}{a(0)}\text{tr}_0 \left(K_0(2\eta)\check R_{0L}'(0)\right),
  \end{equation}
  where $\check R(u) = PR(u)$. The $\check R$-matrix has the property
  \begin{equation}
    \check R'(0) = \frac{a'(0)+c'(0)}{2}+ \frac{b'(0)}{2} \sum_{\alpha=1}^3 J_\alpha \sigma^\alpha\otimes \sigma^\alpha ,
  \end{equation}
  where the anisotropy parameters $J_1,J_2,J_3$ are given by \eqref{eqn:GeneralAP}. The insertion of this expression into \eqref{eqn:TMLogDerIntermediate} leads to \eqref{eqn:TMLogDer} with the boundary terms
  \begin{equation}
    \label{eqn:hBFromK}
    (h_{\text{\tiny B}}^-)_1 = -\frac{a(0)}{2b'(0)}K_1'(0),\quad (h_{\text{\tiny B}}^+)_L = -\frac{1}{4}\text{tr}_0\left(K_0(2\eta)\sum_{\alpha=1}^3 J_\alpha \sigma^\alpha_0 \otimes \sigma^\alpha_L  \right).
   \end{equation}
  The evaluation of the partial trace for $h_{\text{\tiny B}}^+$ is straightforward and leads to the expression given in \eqref{eqn:hB}. To see that $h_{\text{\tiny B}}^-$ is given by the same expression, we first note that \eqref{eqn:param8V} and \eqref{eqn:GeneralAP} lead to 
\begin{equation}\label{eqn:JalphaJ}
  J_\alpha = J 
  \frac{\vartheta_{5-\alpha}(2 \eta, p)}{ \vartheta_{5-\alpha}(0, p) }  \quad\text{for}\quad \alpha=1,2,3,
  \end{equation}
 where $J= (\vartheta_4(0,p^2)/\vartheta_4(2\eta,p^2))^2$. Hence, we obtain
\begin{equation}
   h_{\text{\tiny B}}^- =  -\frac{a(0)\vartheta_1'(0,p)}{2 b'(0)} \sum_{\alpha=1}^3 \frac{\mu_\alpha}{\vartheta_{5-\alpha} (0, p)} \sigma^\alpha 
    =  -\frac{a(0)\vartheta_1'(0,p)}{2 J b'(0)}   \sum_{\alpha=1}^3 \frac{J_\alpha \mu_\alpha}{\vartheta_{5-\alpha} (2 \eta, p)} \sigma^\alpha.
\end{equation}
 It remains to be shown that $a(0) \vartheta_1'(0,p)/(J b'(0))= \vartheta_1(2\eta,p)$, which can be accomplished with the help of identities for Jacobi theta functions \cite{gradshteyn:07}.
\end{proof}
\end{proposition}
An immediate consequence of \eqref{eqn:CommT}, \eqref{eqn:TMLogDer} and $\mathcal T(0) = 2a(0)^{2L}$ is:
\begin{corollary}
  \label{corr:HTComm}
  We have $[H_{\textup{\tiny XYZ}},\mathcal T(u)]=0$ where the XYZ Hamiltonian has the anisotropy parameters \eqref{eqn:GeneralAP} and boundary terms \eqref{eqn:hB}.
\end{corollary}

\paragraph{Supersymmetric eight-vertex model.} We now consider the crossing parameter
\begin{equation}
  \eta = \frac{\pi}{3},
\end{equation}
real $\rho,u$, and $0<p<1$. For this choice, the vertex weights $a,b,c,d$ are real and obey the relation \eqref{eqn:CL8V} that defines the supersymmetric eight-vertex model. The spin chain's anisotropy parameters \eqref{eqn:GeneralAP} coincide with the expressions given in \eqref{eqn:SUSYXYZHamiltonian}, where $0<\zeta<1$ is defined by \eqref{eqn:Zabcd}.
 
 It follows from \cref{corr:HTComm} that the transfer matrix of the eight-vertex model commutes with the Hamiltonian \eqref{eqn:SUSYXYZHamiltonian} provided that the parameters of the $K$-matrix are given by 
\begin{equation}
  \label{eqn:DefMu}
  \mu_1 = \frac{\vartheta_4(\eta,p)}{\vartheta_1(\eta,p)}\frac{2\mathop{\textup{Re}}y}{1+|y|^2}, \quad \mu_2 = \frac{\vartheta_3(\eta,p)}{\vartheta_1(\eta,p)}\frac{2\mathop{\textup{Im}}y}{1+|y|^2}, \quad \mu_3 = \frac{\vartheta_2(\eta,p)}{\vartheta_1(\eta,p)}\frac{1-|y|^2}{1+|y|^2}.
\end{equation}
It is possible to express the corresponding $K$-matrices $K^\pm$ in terms of the vertex weights and the parameter $y$, by means of identities for Jacobi theta functions. We anticipated their expressions in the introduction:
\begin{proposition}
\label{prop:KpmWeights}
For the choice \eqref{eqn:DefMu} the K-matrices $K^\pm$ are given by \eqref{eqn:DefKpm}.
\end{proposition}

In the next proposition, we consider the transfer matrix of the supersymmetric eight-vertex model with these $K$-matrices and with $y$ being a solution of \eqref{eqn:YEqn}. For this case, we show that if $\Lambda_L$, defined \eqref{eqn:DefLambda},  is a transfer-matrix eigenvalue, then its eigenspace is contained in the space of the supersymmetry singlet of the XYZ Hamiltonian.
\begin{proposition}
  \label{prop:LambdaLSinglets}
  Let $L\geqslant 1$, $0<\zeta<1$ and $y$ be a solution of \eqref{eqn:YEqn}. If $|\psi\rangle \in V^L$ obeys 
  \begin{equation}
    \mathcal T|\psi\rangle = \Lambda_L|\psi\rangle,
  \end{equation} 
  where $\Lambda_L$ is given in \eqref{eqn:DefLambda}, then $|\psi\rangle$ is a supersymmetry singlet of the $XYZ$ Hamiltonian \eqref{eqn:SUSYXYZHamiltonian} with $\zeta$ given by \eqref{eqn:Zabcd}.
\begin{proof}
  We use the theta-function parameterisation of the eight-vertex model. It follows from \eqref{eqn:TMLogDer} that $|\psi\rangle$ is an eigenstate of the XYZ Hamiltonian \eqref{eqn:SUSYXYZHamiltonian} for the eigenvalue
  \begin{equation}
    E = -L\left(\frac{a'(0)-c'(0)}{2b'(0)}+1\right) -\frac{a(0)}{4b'(0)}\text{tr}\left(K'(0)K(2\eta)\right).
  \end{equation}
  In the first term on the right-hand side of this equality, we recognise the expression \eqref{eqn:GeneralAP} for the anisotropy parameter $J_3=\frac{1}{2}(\zeta^2-1)$. %
  To compute the second term, we use the parameterisation  \eqref{eqn:KMatrixGeneral} of the $K$-matrix in terms of the parameters $\mu_1,\mu_2, \mu_3$ given by \eqref{eqn:DefMu}, as well as the expression \eqref{eqn:JalphaJ}  for the anisotropy parameters. We have
  \begin{equation}
    \frac{a(0)}{4b'(0)}\text{tr}\left(K'(0)K(2\eta)\right) = \frac{1}{2} \sum_{\alpha=1}^3 J_\alpha  \frac{\vartheta_1^2(2\eta, p)}{\vartheta^2_{5-\alpha}(2\eta,p)} \mu_\alpha^2  = 2\sum_{\alpha=1}^3 \frac{\lambda_\alpha^2}{J_\alpha}.
  \end{equation}
The constants $\lambda_1,\lambda_2,\lambda_3$ are given in \eqref{eqn:SUSYXYZHamiltonian}. We use their explicit expression and the relation \eqref{eqn:YEqn} between $\zeta$ and $y$ to compute $\sum_{\alpha=1}^3\lambda_\alpha^2/J_\alpha = (\zeta^2+4\zeta-1)/8$. This yields the eigenvalue
    \begin{equation}
      E = - \frac{L(3+\zeta^2)}{4} - \frac{\zeta^2+4\zeta-1}{4}.
  \end{equation}
  We conclude that $E$ is the ground-state eigenvalue $E_0$, defined in \eqref{eqn:E0}. It follows from \cref{thm:MainTheorem1} that $|\psi\rangle$ is a supersymmetry singlet.
  \end{proof} 
\end{proposition}

\paragraph{Transformations of the transfer matrix.} The transfer matrix of the supersymmetric eight-vertex model with the $K$-matrices \eqref{eqn:DefKpm} has a simple transformation behaviour under certain spin rotations. Let us write $\mathcal T = \mathcal T(a,b,c,d;y)$, to stress the dependence of the transfer matrix on the vertex weights $a,b,c,d$ and the parameter $y$. We have
\begin{align}
  \label{eqn:TPi}
  \begin{split}
  \mathcal R^{1}(\pi)\mathcal T(a,b,c,d;y)\mathcal R^{1}(-\pi) &= \mathcal T(a,b,c,d;y^{-1}),\\
    \mathcal R^{2}(\pi)\mathcal T(a,b,c,d;y)\mathcal R^{2}(-\pi) &= \mathcal T(a,b,c,d;-y^{-1}),\\
  \mathcal R^{3}(\pi)\mathcal T(a,b,c,d;y)\mathcal R^{3}(-\pi) &= \mathcal T(a,b,c,d;-y),
  \end{split}
\end{align}
which is similar to \eqref{eqn:Pi}. (It is possible to work out the transformation behaviour under rotations by the angle $\theta=\pi/2$, but we will not use it.) We note that since these transformations are unitary, the transfer matrices on the right-hand side of these equalities have the same spectrum as $\mathcal T(a,b,c,d;y)$.

\subsection{The transfer matrix and the supercharges}
\label{sec:TMSUSY}
In this section, we establish a commutation relation between the transfer matrix of the supersymmetric eight-vertex model with open boundary conditions and the supercharge of the supersymmetric open XYZ spin chain. To this end, we first establish local relations between the $R$-matrix of the eight-vertex model, the $K$-matrices, the local supercharge of the XYZ Hamiltonian, and certain auxiliary operators. Second, we combine these local relations with the definition of the transfer matrix to obtain the commutation relation.

\paragraph{Local relations.}%
We follow the strategy of \cite{hagendorf:18} and define two operators $A^\uparrow,A^\downarrow:V \to V\otimes V$. Their action on the basis states $|{\uparrow}\rangle$ and $|{\downarrow}\rangle$ is given by
\begin{align}
  \begin{split}
  A^\uu \ket{\uu} &= d \left( -\frac{c}{a}\ket{\uu\dd} + \ket{\dd\uu}\right), \quad A^\uu\ket{\dd} = c\left( \ket{\uu\uu} - \frac{d}{b}\ket{\dd\dd}\right),\\
  A^\downarrow \ket{\uu} &= c\left( \ket{\dd\dd} - \frac{d}{b}\ket{\uu\uu}\right), \quad A^\downarrow\ket{\dd} = d \left( -\frac{c}{a}\ket{\dd\uu} + \ket{\uu\dd}\right).
   \end{split}
\end{align}
We also define an operator $A_\phi:V \to V\otimes V$ through the following action on the basis states:
\begin{align}
\begin{split}
A_\phi\ket\uu &= (2a+b)\phi_\uu\ket{\uu\uu} + (a+2b)\phi_\dd  \ket{\uu\dd}+ c\phi_\dd \ket{\dd\uu} + d \phi_\dd \ket{\dd\dd},\\
A_\phi\ket\dd &= (2a+b)\phi_\dd\ket{\dd\dd} + (a+2b)\phi_\uu  \ket{\dd\uu}+ c\phi_\uu \ket{\uu\dd} + d \phi_\uu \ket{\uu\uu}.
\end{split}
\end{align}
Here, $\phi_\uu = y(y^2\zeta -1)$ and $\phi_\dd = \zeta-y^2$ are the components of the state $|\phi\rangle$ defined in \eqref{eqn:XYZLocalSupercharge}. We use the operators $A^\uparrow,A^{\downarrow}$, and $A_\phi$ to define the linear combination
\begin{equation}
  \label{eqn:DefA}
  A = (1-y^2\zeta)A^\uparrow +y(y^2-\zeta) A^\downarrow + A_\phi.
\end{equation}

We also need an action of $A,A^\uparrow,A^\downarrow$ and $A_\phi$ on the space $V_0\otimes V^L$.
 To this end, we introduce the following notation: For each operator $B:V\to V\otimes V$ we define $B_0^j:V_0\otimes V^L\to V_0\otimes V^{L+1}$, $j=1,\dots,L+1$ by
\begin{equation}
  B_0^1 = B \otimes \underset{L}{\underbrace{1 \otimes \cdots \otimes 1}}
\end{equation}
and, recursively,
\begin{equation}
B^{j+1}_0 = P_{jj+1} B_0^{j},
\end{equation}
for each $j = 1, \dots,L$. Here, $P_{jj+1}$, $j=1,\dots,L$ denotes the permutation operator acting on the factors $V_j$ and $V_{j+1}$ of the tensor product $V_0\otimes V^{L+1}.$

In the next two lemmas, we establish several relations between the $R$-matrix of the supersymmetric eight-vertex model, the $K$-matrices $K^\pm$ defined in \eqref{eqn:DefKpm}, the local supercharge $\q$ and the operator $A$.
\begin{lemma}
\label{lem:qRA}
For each $i=1,\dots,L$ we have
\begin{subequations}\label{eqn:RRq}
\begin{align}
\label{eqn:RRq1}
R_{0j}R_{0j+1} (1\otimes \q_j)+(a+b) (1\otimes \q_j) R_{0j} &= R_{0j} A_0^{j+1}+A_0^jR_{0j},\\
\label{eqn:RRq2}
R_{0j+1}R_{0j} (1\otimes \q_j) +(a+b) (1\otimes \q_j) R_{0j} &= R_{0j+1} A_0^{j}+A_0^{j+1}R_{0j},
\end{align}
\end{subequations}
if and only if \eqref{eqn:CL8V} holds.
\begin{proof}
 The multiplication of \eqref{eqn:RRq1} from the left by $P_{jj+1}$ yields \eqref{eqn:RRq2} by virtue of $P_{jj+1}\q_j = \q_j$. Hence, it is sufficient to prove \eqref{eqn:RRq1}. 
 
 The key observation is that each of the relations
 \begin{align}
 \begin{split}
   & R_{01}R_{02} (1\otimes 
   (\q^\uparrow)_1) +(a+b)  (1\otimes (\q^\uparrow)_1) R_{01} = R_{01} (A^\uparrow)_0^{2}+(A^\uparrow)_0^1R_{01},\\
    & R_{01}R_{02}  (1\otimes (\q^\downarrow)_1)+(a+b) (1\otimes(\q^\downarrow)_1
     )  R_{01} = R_{01} (A^\downarrow)_0^{2}+(A^\downarrow)_0^1R_{01},\\
    & R_{01}R_{02} (1\otimes (\q_\phi)_1) +(a+b) (1\otimes (\q_\phi)_1) R_{01} = R_{01} (A_\phi)_0^{2}+(A_\phi)_0^1R_{01},
 \end{split}
 \end{align}
 holds if (and only if) the vertex weights obey \eqref{eqn:CL8V}, as follows from a straightforward calculation. Using the definition \eqref{eqn:DefA}, we obtain \eqref{eqn:RRq1} for $j=1$. Its generalisation to $j=2,\dots, L$ is readily obtained through the conjugation with appropriate products of permutation operators.
\end{proof}
\end{lemma}
\begin{lemma}
  \label{lem:RAK}
 The $K$-matrices \eqref{eqn:DefKpm} obey
   \begin{align}
    \label{eqn:RKA1}
    (a+b)A_0^1K_0^- & = R_{01}K_0^- A_0^1,\\
    (a+b)(A_0^1)^{t_0}(K_0^+)^{t_0} & = (R_{01})^{t_0}(K_0^+)^{t_0} (A_0^1)^{t_0},
    \label{eqn:RKA2}
  \end{align}
  if and only \eqref{eqn:CL8V} holds. 
  Here, the superscript $t_0$ denotes the transposition with respect to the auxiliary space.
  \begin{proof}
    The proof is a straightforward calculation.
  \end{proof}
\end{lemma}

\paragraph{The commutation relation.} We now use the \cref{lem:qRA,lem:RAK} to compute a commutation relation between the transfer matrix and the supercharge. This generalises a relation established by Weston and Yang \cite{weston:17} for the six-vertex model, corresponding to $d=0,\,y=0$.
\begin{proposition}
  \label{prop:TQCommRel}
  If \eqref{eqn:CL8V} holds and the $K$-matrices $K^\pm$ are given by \eqref{eqn:DefKpm} then
\begin{equation}
   \label{eqn:CommRelTQ}
   \mathcal T \Q = (a+b)^2\Q \mathcal T.
\end{equation}
\begin{proof}
  First, we evaluate a commutator between the transfer matrix and the local supercharge $\q_j$. To this end, we use 
\begin{align}
R_{0k}(1\otimes \q_j) & = (1\otimes \q_j)  R_{0k}, &&\text{if } 1\leqslant k<j \leqslant L, \\
R_{0k}(1\otimes \q_j) & = (1\otimes \q_j) R_{0k-1}, &&\text{if } 1\leqslant j< k-1 \leqslant L-1.
\end{align}
We apply them together with \cref{lem:qRA} to obtain
  \begin{align*}
    \mathcal T\q_j-(a+b)^2\q_j \mathcal T = \text{tr}_0\left(K_0^+U_{0,[1,L+1]}K_0^-\bar U_{0,[1,j-1]}
    \left(R_{0j} A_0^{j+1}+A_0^jR_{0j}\right)
    \bar U_{0,[j+1,L]}
    \right)\\
    -(a+b)\text{tr}_0\left(K_0^+U_{0,[j+2,L+1]}\left(R_{0j+1}A_0^j+A_0^{j+1}R_{0j}\right)U_{0,[1,j-1]}K_0^-
    \bar U_{0,[1,L]}
    \right),
  \end{align*}
  for $j=1,\dots,L$.
  
  Second, we take an alternating sum of these equalities and find
  \begin{align}
    \label{eqn:TQAfterSum}
    \begin{split}
    \mathcal T \Q - (a+b)^2\Q \mathcal T = \, & \text{tr}_0\left(K_0^+U_{0,[2,L+1]}\left((a+b)A_0^1K_0^- - R_{01}K_0^-A_0^1\right)\bar U_{0,[1,L]}
    \right)\\
    & +(-1)^L\left(\text{tr}_0\left(K_0^+ R_{0L+1}\mathcal U A_0^{L+1}\right)-(a+b)\text{tr}_0\left(K_0^+A_0^{L+1} \mathcal U \right)\right),
    \end{split}
  \end{align}
  where we used the shorthand notation $\mathcal U = U_{0,[1,L]}K_0^- \bar U_{0,[1,L]}$.
  The relation \eqref{eqn:RKA1} implies that the first term on the right-hand side of \eqref{eqn:TQAfterSum} vanishes. To evaluate the second term, we compute
  \begin{align}
    \begin{split}
    &\text{tr}_0\left(K_0^+ R_{0L+1}\mathcal U A_0^{L+1}\right) = \text{tr}_0 \left(\mathcal U^{t_0}(R_{0L+1})^{t_0} 
    (K_0^+)^{t_0}(A_{0}^{L+1})^{t_0}\right)\\
    & = \, (a+b)\text{tr}_0 \left(\mathcal U^{t_0}(A_{0}^{L+1})^{t_0}(K_0^+)^{t_0}\right) =(a+b)\text{tr}_0\left(K_0^+A_0^{L+1} \mathcal U \right). 
    \end{split}
  \end{align}
  To establish this equality, we used the invariance of the trace under matrix transposition and applied the identity  $(R_{0L+1})^{t_0} (K_0^+)^{t_0}(A_{0}^{L+1})^{t_0}= (a+b)(A_{0}^{L+1})^{t_0}(K_0^+)^{t_0}$, which follows from \eqref{eqn:RKA2} after an appropriate multiplication with permutation operators. Hence, we conclude that the second term on the right-hand side of \eqref{eqn:TQAfterSum} vanishes, too. 
\end{proof}
\end{proposition}

\subsection{The eigenvalue}
\label{sec:TMEV}

In this section, we prove \cref{thm:MainTheorem2}. We prepare its proof by establishing a few auxiliary results. Below, we denote by $\mathcal T$ the transfer matrix of the supersymmetric eight-vertex model on a strip with $L\geqslant 1$ vertical lines, the $K$-matrices $K^\pm$ defined in \eqref{eqn:DefKpm} and $t=\pi/6$. We compute the action of this transfer matrix on the supersymmetry singlet $|\Psi_L\rangle$ defined in \eqref{eqn:CohomologyRepresentation}. This singlet is an eigenstate of $H$, and thus of $H_{\textup{\tiny XYZ}}$. Therefore, it is an eigenstate of $\mathcal T$. The eigenvalue $\Lambda_L$ can be obtained as
\begin{equation}
  \Lambda_L = \frac{\langle \Psi_L|\mathcal T |\Psi_L\rangle}{\langle \Psi_L|\Psi_L\rangle}.
  \label{eqn:EVRitz}
\end{equation}
We evaluate this quotient by using the following proposition, whose proof is identical to the one of Proposition 3.4 in \cite{hagendorf:18}.
\begin{proposition}
\label{prop:QExactOp}
Let $L\geqslant 1$ and $|\psi\rangle\in V^L$ be a supersymmetry singlet with the decompositions $|\psi\rangle = |\phi\rangle + \Q|\gamma\rangle$ (or $|\psi\rangle = |\phi\rangle$ for $L=1$) and $|\psi\rangle = |\phi'\rangle + \Q^\dagger |\gamma'\rangle$. Let $\mathcal A$ be an operator defined on $V^L$ for each $L\geqslant 1$ that obeys the commutation relation
\begin{equation}
  \mathcal A \Q = \lambda \Q \mathcal A,
\end{equation}
with non-zero $\lambda$. Then we have 
\begin{equation}
  \langle \psi|\mathcal A|\psi\rangle = \langle \phi'|\mathcal A |\phi\rangle.
\end{equation}
\end{proposition}

It follows from \cref{prop:TQCommRel} that if $a+b\neq 0$ then we may apply \cref{prop:QExactOp} with $\mathcal A = \mathcal T$ and $\lambda=(a+b)^2$ to evaluate the matrix element $\langle \Psi_L|\mathcal T |\Psi_L\rangle$. Furthermore, we compute the square norm $\langle \Psi_L|\Psi_L\rangle$ with the help of this proposition for $\mathcal A = 1$ and $\lambda=1$. The resulting expressions depend on the choice of the decompositions of $|\Psi_L\rangle$. First, using \eqref{eqn:CohomologyRepresentation} and \eqref{eqn:HomologyRepresentation}, we have
\begin{equation}
  \label{eqn:ThetaL1}
  \Lambda_L = \frac{\left(\langle w_+|^{\otimes L}\right)\mathcal T\left(|v_+\rangle^{\otimes L}\right)}{\langle w_+|v_+\rangle^L},
\end{equation}
for each $L\geqslant 1$.
Second, using the alternative representations \eqref{eqn:AlternativeDecomp1} and \eqref{eqn:AlternativeDecomp2}, we find
\begin{equation}
  \label{eqn:ThetaL2}
  \Lambda_L = \frac{\left(\langle \alpha|\otimes \langle w_+|^{\otimes (L-2)}\right)\mathcal T\left(|\chi\rangle\otimes |v_+\rangle^{\otimes (L-2)}\right)}{\langle \alpha|\chi\rangle\langle w_+|v_+\rangle^{L-2}},
\end{equation}
for each $L\geqslant 2$. These two relations still hold if $a+b=0$. Indeed, the eigenvalues of a matrix are continuous functions of its entries \cite{meyer:00}. Hence, $\Lambda_L$ is a continuous function of $a,b,c,d$.

We exploit \eqref{eqn:ThetaL1} and \eqref{eqn:ThetaL2} to establish a recurrence relation for the eigenvalue $\Lambda_L$. To this end, we need the following two lemmas:
\begin{lemma}
\label{lem:MagicIdentity1}
  For $t=\pi/6$, the $K$-matrices \eqref{eqn:DefKpm} obey
\begin{equation}
  \label{eqn:ReccK1}
\frac{\bra{w_+} \textup{tr}_0\left( K^+_0R_{01}K^-_0R_{01}\right) \ket{v_+} }{\langle w_+ \ket{v_+}}
 = (a+b)^2  \textup{tr}(K^+ K^-).
\end{equation}
\begin{proof}
By virtue of \cref{lem:VWPolynomial}, it is sufficient to show that
\begin{equation}
  I = \bra{\bar w_+} \textup{tr}_0\left( K^+_0R_{01}K^-_0R_{01}\right) \ket{\bar v_+} -\langle \bar w_+ |{\bar v_+}\rangle(a+b)^2  \textup{tr}(K^+ K^-)
\end{equation}
vanishes. This difference is a rational expression of the vertex weights $a,b,c,d$, $\zeta$ and the parameter $y$. Using the relations \eqref{eqn:CL8V}, \eqref{eqn:YEqn} and \eqref{eqn:Zabcd}, we find after some algebra, that is indeed zero. 
\end{proof}
\end{lemma}
\begin{lemma}
\label{lem:MagicIdentity2}
 For $t=\pi/6$, the matrix $K^-$, defined in \eqref{eqn:DefKpm}, obeys
\begin{equation}
  \frac{(1\otimes \langle \alpha|)R_{02}R_{01}K_0^- R_{01}R_{02}(1\otimes |\chi\rangle)}{\langle \alpha|\chi\rangle} = (a+b)^4K_0^-.%
\end{equation}
\begin{proof}
   By virtue of \cref{lem:ChiAlphaPolynomial}, the equality holds if the $2\times 2$ matrix
  \begin{equation}
   \bar I= (1\otimes \langle \bar \alpha|)R_{02}R_{01}K_0^- R_{01}R_{02}(1\otimes |\bar \chi\rangle) -(a+b)^4\langle \bar \alpha|\bar \chi\rangle K_0^-
\end{equation}
vanishes. Its entries are rational expressions of the vertex weights $a,b,c,d$, $\zeta$ and the parameter $y$. As above, we use \eqref{eqn:CL8V}, \eqref{eqn:YEqn} and \eqref{eqn:Zabcd} to show that its entries are indeed zero. 
\end{proof}
\end{lemma}

\begin{proof}[Proof of \cref{thm:MainTheorem2}] According to \cref{prop:LambdaLSinglets}, if $L\geqslant 1$, $0<\zeta<1$, and if $y$ is a solution of \eqref{eqn:YEqn}, then any solution $|\psi\rangle$ of $\mathcal T|\psi\rangle = \Lambda_L|\psi\rangle$ is a supersymmetry singlet. This observation does, however, not guarantee that $\Lambda_L$ is an eigenvalue of the transfer matrix because a solution of the eigenvalue problem might not exist. To see that it is an eigenvalue, we thus evaluate the action transfer matrix on $|\Psi_L\rangle$. To this end, we use \eqref{eqn:EVRitz}.

First, we consider $t=\pi/6$ and hence $y=y_0$, where $y_0$ is the unique real solution of \eqref{eqn:YEqn} with $0<y<1$. We suppose $L\geqslant 3$, and use the definition of the transfer matrix to rewrite \eqref{eqn:ThetaL2} as
\begin{equation}
  \Lambda_L = \frac{ \left(\langle \alpha|\otimes \langle w_+|^{\otimes (L-2)}\right) \text{tr}_0 \left( K_0^+ U_{0,[3, L]} R_{02} R_{01} K_0^- R_{01} R_{02} \bar U_{0,[3,L]}\right) \left(|\chi\rangle\otimes |v_+\rangle^{\otimes (L-2)}\right) }{\langle \alpha|\chi\rangle\langle w_+|v_+\rangle^{L-2}}.\nonumber
\end{equation}
We apply \cref{lem:MagicIdentity2} on the right-hand side of this equality and obtain, after a redefinition of labels, the expression
\begin{equation}
  \Lambda_L = (a+b)^4\frac{\langle w_+|^{\otimes (L-2)} \text{tr}_0 \left(  K_0^+ U_{0,[1, L-2]} K_0^-  \bar U_{0,[1,L-2]} \right)|v_+\rangle^{\otimes (L-2)}}{\langle w_+|v_+\rangle^{L-2}}.
  \end{equation}
Now, we use \eqref{eqn:ThetaL1} to recognise on the right-hand side of this equality $\Lambda_{L-2}$. Therefore, we have the recurrence relation
\begin{equation}
  \Lambda_L = (a+b)^4 \Lambda_{L-2}.
\end{equation}
To solve this recurrence, we compute the eigenvalues $\Lambda_L$ for $L=1,2$. They immediately follow from \cref{lem:MagicIdentity1,lem:MagicIdentity2}. We find
\begin{align}
\begin{split}
\Lambda_1&=\frac{\bra{w_+} \textup{tr}_0\left( K^+_0R_{01}K^-_0R_{01}\right) \ket{v_+}}{\langle w_+ \ket{v_+}} = (a+b)^2  \textup{tr}(K^+ K^-),\\
\Lambda_2&=\frac{\bra{\alpha} \text{tr}_0\left(K_0^+ R_{02} R_{01} K_0^- R_{01} R_{02}\right) \ket\chi }{\langle \alpha|\chi\rangle} =(a+b)^4 \text{tr}(K^+ K^-).
\end{split}
\end{align}
The solution of the recurrence relation with these initial conditions leads to the eigenvalue $\Lambda_L = (a+b)^{2L}\text{tr}(K^+K^-)$, for each $L\geqslant 1$. The eigenspace of $\Lambda_L$ is by construction the space spanned by the supersymmetry singlet $|\Psi_L\rangle$. It is one-dimensional. Therefore $\Lambda_L$ is non-degenerate.

Second, we consider the other real solutions $y=y_\alpha,\,\alpha=1,2,3,$ of \eqref{eqn:YEqn}. It follows from \eqref{eqn:TPi} that the corresponding transfer matrix has the property
  \begin{equation}
    \mathcal T(a,b,c,d;y_\alpha)=\mathcal R^\alpha(-\pi)\mathcal T(a,b,c,d;y_0)\mathcal R^\alpha(\pi).
  \end{equation}
  The two transfer matrices in this equality are related by a unitary transformation. Therefore, they have the same eigenvalues with the same degeneracies. Hence, the transfer matrix possesses the eigenvalue $\Lambda_L$ in this case, too. Its eigenspace is the span of the supersymmetry singlet $|\Psi_L^\alpha\rangle$, defined in the proof of \cref{thm:MainTheorem1}.
\end{proof}

\section{The largest eigenvalue}
\label{sec:LargestEV}

The relation \eqref{eqn:CL8V} admits positive solutions. Indeed, using the parameterisation \eqref{eqn:param8V}, we have $a,b,c,d>0$ if $\rho>0,\ \eta=\pi/3,\ 0<u<\pi/3,$ and $0<p<1$. We now prove that in this case, $\Lambda_L$ is the largest eigenvalue of the transfer matrix $\mathcal T$ of the supersymmetric eight-vertex model with the $K$-matrices \eqref{eqn:DefKpm} and $y$ a solution of \eqref{eqn:YEqn}.

The proof is based on the Perron-Frobenius theorem for positive matrices and its variant for non-negative matrices. We use certain concepts from Perron theory and refer to the book \cite{meyer:00} for details. We only recall that $|\psi\rangle \in V^L$ is called a \textit{Perron vector} if all its components are positive and its norm is one.

\begin{proposition}
\label{prop:PositivePsi}
For each $L\geqslant 1$, there is a constant $C_L$ such that $|\Psi_L'\rangle = C_L|\Psi_L\rangle$ is a Perron vector.\begin{proof}
First, we note that for all $(p,t)\in \bar{\mathbb D}$ with $t=\pi/6$, the off-diagonal matrix elements of the Hamiltonian $H_{\textup{\tiny XYZ}}$ are zero or negative. Hence, there is a real number $\lambda$ such that the matrix $\lambda-H_{\textup{\tiny{XYZ}}}$ has a positive diagonal and non-negative off-diagonal entries.

Second, we note that the action of $\lambda-H_{\textup{\tiny{XYZ}}}$ on any basis state $|s_1\cdots s_L\rangle$ of $V^L$ leads to a linear combination of basis states that are obtained from $|s_1\cdots s_L\rangle$ by \textit{(i)} flipping pairs adjacent aligned spins, \textit{(ii)} exchanging pairs of adjacent anti-aligned spins, \textit{(iii)} flipping the spin on the first or last site or \textit{(iv)} leaving the basis state unchanged. The coefficients of this linear combination are positive. The repeated application of the operations \textit{(i)}-\textit{(iv)} allows one to generate any basis state from $|s_1\cdots s_L\rangle$. We conclude that there is an integer $m> 0$ such that $(\lambda-H_{\textup{\tiny XYZ}})^m$ has positive entries. Hence, $\lambda-H_{\textup{\tiny XYZ}}$ is a non-negative irreducible matrix.

Third, we apply the Perron-Frobenius theorem to the matrix $\lambda-H_{\textup{\tiny XYZ}}$. It implies that its largest eigenvalue is non-degenerate and that the corresponding eigenspace is spanned by a Perron vector $|\Psi_L'\rangle$. By \cref{thm:MainTheorem1} this largest eigenvalue is $\lambda-E_0$, and the eigenspace spanned by $|\Psi_L\rangle$. Hence, there must be a constant $C_L$ such that $|\Psi_L'\rangle = C_L|\Psi_L\rangle$.
\end{proof}
\end{proposition}

\begin{proposition}
  \label{prop:PositiveT}
  For each $L\geqslant 1$, positive vertex weights $a,b,c,d$, $0<\zeta<1$ and real $0< y < 1$, the transfer matrix of the supersymmetric eight-vertex model on a strip of length $L$ with the $K$-matrices $K^\pm$ defined in \eqref{eqn:DefKpm} is a positive matrix.
  \begin{proof}
 Let $V_0,V_{\bar{0}}=V$ be two copies of the single-spin Hilbert space. For each $s,\bar{s}\in\{\uparrow,\downarrow\}$, we define an operator $C^{s\bar{s}}:V_0\otimes V_{\bar{0}}\to V_0\otimes V_{\bar{0}}$ by
  \begin{equation}
    C^{s\bar{s}} = \left(1\otimes 1 \otimes \langle \bar{s}|\right)R_{01}(R_{\bar{0}1})^{t_{\bar0}} 
    \left(1\otimes 1 \otimes |s\rangle\right).
  \end{equation}
  Its entries are non-negative. A direct calculation shows that for all $s,\bar{s}\in\{{\uparrow,\downarrow}\}$ and each $|p\rangle \in \{|{\uparrow\uparrow}\rangle,|{\uparrow\downarrow}\rangle\}$ there is a unique $|\bar p\rangle \in \{|{\uparrow\uparrow}\rangle,|{\uparrow\downarrow}\rangle\}$, depending on $s,\bar{s}$, such that
   $\langle \bar p|C^{s\bar{s}}|p\rangle > 0$. Moreover, we define two states $|k^\pm\rangle\in V_0\otimes V_{\bar{0}}$ through their components, given by
  \begin{equation}
    \langle s\bar{s}|k^\pm\rangle = \langle \bar{s}|K^\pm |s\rangle,
  \end{equation}
  for all $s,\bar{s}\in\{{\uparrow,\downarrow}\}$. These components are positive.
  
  For each pair of basis states $|s_1\cdots s_L\rangle,|\bar s_1\cdots \bar s_L\rangle$, we write the matrix elements of the transfer matrix in terms of these operators and states:
  \begin{equation}
    \langle\bar s_1\cdots \bar s_L|\mathcal T|s_1\cdots s_L\rangle = \langle k^+|    C^{s_L\bar s_L}\cdots C^{s_1\bar s_1}|k^-\rangle.
  \end{equation}
To investigate this matrix element, we use the identity
  \begin{equation}
    \sum_{|\omega\rangle \in \Omega}|\omega\rangle\langle \omega| = 1,
  \end{equation}
  where $\Omega=\{|{\uparrow\uparrow}\rangle,|{\uparrow\downarrow}\rangle,|{\downarrow\uparrow}\rangle,|{\downarrow\downarrow}\rangle\}$ denotes the canonical basis of $V^2$. It allows us to write
  \begin{equation}
  \bra{\bar{s}_1\dots \bar{s}_L} \mathcal T \ket{s_1\dots s_L}= \sum_{|\omega_0\rangle,\dots,|\omega_{L}\rangle \in\Omega}  \langle k^+|\omega_{L}\rangle \Biggl( \prod_{j=1}^{L}  \bra{\omega_{j}}C^{s_{j}\bar{s}_{j} } \ket{\omega_{j-1}} \Biggr)  \langle \omega_0|k^-\rangle.
  \end{equation}
  Each term inside the sum of the right-hand side is a product of non-negative factors. To show that the sum is positive, it is therefore sufficient to find a single choice for $|\omega_0\rangle,\dots,|\omega_{L}\rangle$ that yields a positive term. We determine such a choice by iteration. First, we set $|\omega_0\rangle = |p_0\rangle=|{\uparrow\uparrow}\rangle$. Second, we choose the unique state $|\omega_1\rangle = |p_1\rangle \in\{|{\uparrow\uparrow}\rangle,|{\uparrow\downarrow}\rangle\}$ such that $\langle \omega_1|C^{s_1\bar s_1}|\omega_0\rangle=\langle p_1|C^{s_1\bar s_1}|p_0\rangle>0$. Next, we iterate this step and determine for each $i=2,\dots,L$ the unique $|\omega_i\rangle = |p_i\rangle\in\{|{\uparrow\uparrow}\rangle,|{\uparrow\downarrow}\rangle\}$ such that $\langle \omega_{i}|C^{s_{i}\bar s_{i}}|\omega_{i-1}\rangle=\langle p_{i}|C^{s_i\bar s_i}|p_{i-1}\rangle > 0$. The term corresponding to this choice is a lower boundary for the sum:
  \begin{equation}
  \bra{\bar{s}_1\dots \bar{s}_L} \mathcal T \ket{s_1\dots s_L} \geqslant \langle k^+|p_{L}\rangle \Biggl( \prod_{j=1}^{L}  \bra{p_{j}}C^{s_{j}\bar{s}_{j} } \ket{p_{j-1}} \Biggr)  \langle p_0|k^-\rangle.
  \end{equation}
Each factor of the product on the right-hand side of this equality is positive. Hence, the matrix element is positive.
  \end{proof}
 
\end{proposition}

\begin{proof}[Proof of \cref{thm:MainTheorem3}]
  First, let $y=y_0$ be the unique solution of the equation \eqref{eqn:YEqn} with $0<y<1$. We denote by $\Lambda_L'$ the largest eigenvalue of the transfer matrix $\mathcal T = \mathcal T(a,b,c,d;y_0)$ of the supersymmetric eight-vertex model with the $K$-matrices \eqref{eqn:DefKpm} and positive vertex weights $a,b,c,d>0$. By \cref{prop:PositiveT}, $\mathcal T$ is a positive matrix. The Perron-Frobenius theorem states that the eigenspace of $\Lambda_L'$ is one-dimensional and spanned by a Perron vector, and that no other eigenspace contains a Perron vector. We have $\mathcal T|\Psi_L'\rangle = \Lambda_L|\Psi_L'\rangle$, where $|\Psi_L'\rangle$ is the Perron vector of \cref{prop:PositivePsi}. Hence, $\Lambda_L'=\Lambda_L$.
  
  Second, let $y=y_\alpha,\,\alpha=1,2,3,$ be another solution of \eqref{eqn:YEqn}. We follow the reasoning of the proof of \cref{thm:MainTheorem2}. The transfer matrix has the property
  \begin{equation}
    \mathcal T(a,b,c,d;y_\alpha)=\mathcal R^\alpha(-\pi)\mathcal T(a,b,c,d;y_0)\mathcal R^\alpha(\pi).
  \end{equation}
  The two transfer matrices in this equality are related by a unitary transformation. Therefore, they have the same spectrum and, hence, the same largest eigenvalue $\Lambda_L$.  \end{proof}

\paragraph{The free energy.} Up to an irrelevant factor, the free energy per pairs of horizontal lines of the eight-vertex model on a strip is given by the logarithm of the largest eigenvalue of its transfer matrix. For large $L$, it is expected to take the form
\begin{equation}
  -\ln \Lambda_L = 2L f + f_{\textup{\tiny B}}+ O(L^{-1}),
  \label{eqn:LogLambda}
\end{equation}
where $f$ is the bulk free energy per site, and $f_{\textup{\tiny B}}$ the boundary free energy. The bulk free energy per site is known from Baxter's work \cite{baxterbook}. As for $f_{\textup{\tiny B}}$, however, we are not aware of an explicit formula for general vertex weights and boundary conditions in the literature.

In the case studied in this article, it is trivial to compute the expansion \eqref{eqn:LogLambda}, because we explicitly know $\Lambda_L$ for each $L\geqslant 1$. We obtain
\begin{equation}
  f = -\ln(a+b), \quad f_{\textup{\tiny B}} = - \ln \textup{tr}(K^+K^-).
\end{equation}
The finite-size corrections $O(L^{-1})$ are absent. We note that $f=-\ln(a+b)$ matches Baxter's results \cite{baxterbook}. 

\section{Conclusion}
\label{sec:Conclusion}

In this article, we studied the Hamiltonian of an open XYZ spin chain with a lattice supersymmetry and the corresponding transfer matrix of the eight-vertex model on a strip. We showed that if the parameters of the Hamiltonian are carefully adjusted then its ground states are supersymmetry singlets. The space of supersymmetry singlets is an eigenspace of the transfer matrix. We computed the corresponding eigenvalue with the help of a commutation relation between the supercharge and the transfer matrix.
 For positive vertex weights, we showed that it is the largest eigenvalue. The techniques that we used to prove these results rely on supersymmetry, (co)homology, integrability and the Perron-Frobenius theorem.

We conclude this article with a conjecture that generalises the transfer-matrix eigenvalue to the inhomogeneous eight-vertex model on the strip with $L\geqslant 1$ vertical lines. Its transfer matrix is
\begin{equation}
  \label{eqn:InhTM}
  \mathcal T(u|u_1,\dots,u_L) =\text{tr}_0\left(K_0^+(u)U_{0,[1,L]}(u|u_1,\dots,u_L)K_0^-(u)
  \bar U_{0,[1,L]}(u|u_1,\dots,u_L)\right),
\end{equation}
where $K^-(u) = K(u)$ and $K^+(u)=K(u+2\eta)$, and
\begin{align}
  \begin{split}
  U_{0,[1,L]}(u|u_1,\dots,u_L)&= R_{0L}(u+u_L)\cdots R_{01}(u+u_1),\\
\bar U_{0,[1,L]}(u|u_1,\dots,u_L)&= R_{01}(u-u_1)\cdots R_{0L}(u-u_L).
  \end{split}
\end{align}
Here, $u_1,\dots,u_L$ are the so-called inhomogeneity parameters.
\begin{conjecture}
  Let $\eta=\pi/3$, $K(u)$ be the $K$-matrix \eqref{eqn:KMatrixGeneral} with the coefficients \eqref{eqn:DefMu}, evaluated at $t=\pi/6$, then the transfer matrix \eqref{eqn:InhTM} possesses the eigenvalue
  \begin{equation}
    \Lambda_L= \textup{tr} (K^+(u)K^-(u))\prod_{j=1}^L \left(a(u+u_j)+b(u+u_j)\right)\left(a(u-u_j)+b(u-u_j)\right).
  \end{equation}
\end{conjecture}
We checked this conjecture numerically for small $L$ in the trigonometric limit $p\to 0$. Furthermore, we checked that it is compatible with functional equations, obeyed by the transfer matrix, and simplifications that occur for certain specialisations of the spectral parameter $u$.

We note that a similar conjecture exists for the transfer matrix of the inhomogeneous eight-vertex model with $\eta=\pi/3$ and periodic boundary conditions \cite{razumov:10,zinnjustin:13}. Both these conjectures remain to be proven. Their proof is of interest since the inhomogeneous models allow one to investigate the properties of the corresponding eigenvectors rigorously. For periodic boundary conditions, Zinn-Justin initiated this rigorous investigation in \cite{zinnjustin:13}.

\subsection*{Acknowledgements}

This work is supported by the Belgian Excellence of Science (EOS) initiative through the project 30889451 PRIMA -- Partners in Research on Integrable Systems and Applications. We thank Gilles Parez for his comments on the manuscript.


\begin{thebibliography}{10}

\bibitem{stroganov:01}
Y.G.~{Stroganov},
\newblock {\em {The importance of being odd}},
\newblock J. Phys. A: Math. Gen. {\textbf{34}} (2001)   L179.

\bibitem{razumov:10}
A.V.~{Razumov} and Y.G.~{Stroganov},
\newblock {\em {A possible combinatorial point for the XYZ spin chain}},
\newblock Theor. Math. Phys. {\textbf{164}} (2010)   977.

\bibitem{mangazeev:10}
V.V.~{Mangazeev} and V.V.~{Bazhanov},
\newblock {\em {The eight-vertex model and Painlev{\'e} VI equation II:
  eigenvector results}},
\newblock J. Phys. A: Math. Theor. {\textbf{43}} (2010) 085206.

\bibitem{zinnjustin:13}
P.~Zinn-Justin,
\newblock {\em {Sum Rule for the Eight-Vertex Model on Its Combinatorial
  Line}},
\newblock in K.~Iohara, S.~Morier-Genoud  and B.~R{\'e}my (eds.), {\em
  Symmetries, Integrable Systems and Representations}, {\bf 40}, 599, Springer London (2013).

\bibitem{baxter:89}
R.J.~Baxter,
\newblock {\em Solving models in statistical mechanics},
\newblock Adv. Stud. Pure Math. {\textbf{19}} (1989)   95.

\bibitem{fabricius:05}
K.~{Fabricius} and B.M.~{McCoy},
\newblock {\em {New Developments in the Eight Vertex Model II. Chains of Odd
  Length}},
\newblock J. Stat. Phys. {\textbf{120}} (2005)   37.

\bibitem{rosengren:16}
H.~Rosengren,
\newblock {\em Elliptic pfaffians and solvable lattice models},
\newblock J. Stat. Mech. (2016)  P083106.

\bibitem{bazhanov:05}
V.V.~{Bazhanov} and V.V.~{Mangazeev},
\newblock {\em {Eight-vertex model and non-stationary Lam{\'e} equation}},
\newblock J. Phys. A: Math. Gen. {\textbf{38}} (2005)   L145.

\bibitem{bazhanov:06}
V.V.~{Bazhanov} and V.V.~{Mangazeev},
\newblock {\em {The eight-vertex model and Painlev{\'e} VI}},
\newblock J. Phys. A: Math. Gen. {\textbf{39}} (2006)   12235.

\bibitem{rosengren:13}
H.~{Rosengren},
\newblock \textit{Special polynomials related to the supersymmetric
  eight-vertex model. I. Behaviour at cusps},
\newblock arXiv:1305.0666 (2013).

\bibitem{rosengren:13_2}
H.~{Rosengren},
\newblock \textit{Special polynomials related to the supersymmetric
  eight-vertex model. II. Schr\"odinger equation},
\newblock arxiv:1312.5879 (2013).

\bibitem{rosengren:14}
H.~{Rosengren},
\newblock \textit{Special polynomials related to the supersymmetric
  eight-vertex model. III. Painlev\'e VI equation},
\newblock arxiv:1405.5318 (2014).

\bibitem{rosengren:15}
H.~{Rosengren},
\newblock {\em {Special Polynomials Related to the Supersymmetric Eight-Vertex
  Model: A Summary}},
\newblock Comm. Math. Phys. {\textbf{340}} (2015)   1143.

\bibitem{fendley:10}
P.~{Fendley} and C.~{Hagendorf},
\newblock {\em {Exact and simple results for the XYZ and strongly interacting
  fermion chains}},
\newblock J. Phys. A: Math. Theor. {\textbf{43}} (2010) 402004.

\bibitem{hagendorf:12}
C.~Hagendorf and P.~Fendley,
\newblock {\em {The eight-vertex model and lattice supersymmetry}},
\newblock J. Stat. Phys. {\textbf{146}} (2012) 1122.

\bibitem{hagendorf:13}
C.~Hagendorf,
\newblock {\em {Spin chains with dynamical lattice supersymmetry}},
\newblock J. Stat. Phys. {\textbf{150}} (2013) 609.

\bibitem{hagendorf:18}
C.~Hagendorf and J.~Li\'enardy,
\newblock {\em {On the transfer matrix of the supersymmetric eight-vertex
  model. I. Periodic boundary conditions}},
\newblock J. Stat. Mech. (2018) 033106.

\bibitem{witten:82}
E.~Witten,
\newblock {\em {Constraints on supersymmetry breaking}},
\newblock Nucl. Phys. {\textbf{B202}} (1982) 253.

\bibitem{hagendorf:17}
C.~Hagendorf and J.~Li\'enardy,
\newblock {\em Open spin chains with dynamic lattice supersymmetry},
\newblock J. Phys. A: Math. Theor. {\textbf{50}} (2017) 185202.

\bibitem{sklyanin:88}
E.K.~Sklyanin,
\newblock {\em Boundary conditions for integrable quantum systems},
\newblock J. Phys. A: Math. Gen. {\textbf{21}}
  {\textbf{10}} (1988) 2375.

\bibitem{cao:13_2}
J.~Cao, W.-L.~Yang, K.~Shi  and Y.~Wang,
\newblock {\em Off-diagonal bethe ansatz solutions of the anisotropic spin-$\frac{1}{2}$
  chains with arbitrary boundary fields},
\newblock Nucl. Phys. \textbf{B877} (2013)   152.

\bibitem{faldella:14}
S.~Faldella and G.~Niccoli,
\newblock {\em {SOV approach for integrable quantum models associated with
  general representations on spin-1/2 chains of the 8-vertex reflection
  algebra}},
\newblock J. Phys. A: Math. Theor. {\textbf{47}} (2014) 115202.

\bibitem{gradshteyn:07}
I.S.~Gradshteyn and I.M.~Ryzhik,
\newblock {\em Table of Integrals, Series and Products},
\newblock Elsevier (2007).

\bibitem{whittaker:27}
E.T.~Whittaker and G.N.~Watson,
\newblock {\em A course of modern analysis},
\newblock Cambridge University Press (1927).

\bibitem{nehari:82}
Z.~Nehari,
\newblock {\em Conformal mapping},
\newblock Dover Publications (1982).

\bibitem{masson:08}
T.~Masson,
\newblock {\em {Introduction aux (Co)Homologies}},
\newblock \'Editions Hermann, Paris (2008).

\bibitem{baxterbook}
R.J.~Baxter,
\newblock {\em {Exactly solved models in statistical mechanics}},
\newblock London Academic (1982).

\bibitem{inami:94}
T.~Inami and H.~Konno,
\newblock {\em {Integrable XYZ spin chain with boundaries}},
\newblock J. Phys. A : Math. Gen. {\textbf{27}} (1994)   L913.

\bibitem{hou:95}
B.-Y.~{Hou}, K.-J.~{Shi}, H.~{Fan}  and Z.-X.~{Yang},
\newblock {\em Solution of reflection equation},
\newblock Comm. Theor. Phys. {\textbf{23}} (1995)   163.

\bibitem{vega:94}
H.J.~{de Vega} and A.~{Gonzalez-Ruiz},
\newblock {\em Boundary K-matrices for the XYZ, XXZ and XXX spin chains},
\newblock  J. Phys. A : Math. Gen. {\textbf{27}} (1994) 6129.

\bibitem{weston:17}
R.~Weston and J.~Yang,
\newblock {\em {Lattice supersymmetry in the open XXZ model: an algebraic Bethe
  Ansatz analysis}},
\newblock J. Stat. Mech. (2017) P123104.

\bibitem{meyer:00}
C.D.~Meyer,
\newblock {\em {Matrix Analysis and Applied Linear Algebra}},
\newblock SIAM (2000).

\end{thebibliography}
\end{document}